\documentclass[manyauthors]{fundam} 
\usepackage[utf8]{inputenc}

\usepackage{tikz}
\usetikzlibrary{shapes}
\usetikzlibrary{arrows}

\usepackage{amsfonts}
\usepackage{amsmath}
\usepackage{graphics,graphicx} 
\usepackage{xcolor}
\usepackage{algorithmic} 
\usepackage{algorithm}
\usepackage{multirow}
\usepackage{placeins}

\newcommand\ie{{\emph{i.e.}}}

\newcommand{\verifcar}{{\sc Verif\-Car}} 
\newcommand{\uppaal}{{\sc Uppaal}}
\newcommand{\zinc}{{\sc Zinc}}
\newcommand{\model}{{\sc MAPT}}
\newcommand{\gmodel}{{\sc G-MAPT}}
\newcommand{\cavs}{{\sc CAVs}}
\newcommand{\cav}{{\sc CAV}}
\newcommand{\acygra}{{\sc DAG}}
\newcommand{\ctl}{{\sc CTL}}
\newcommand{\lcm}{{\sf lcm}}
\newcommand{\lcmE}{{\sf lcm(E)}} 
\newcommand{\defeq}{\stackrel{\mbox{{\tiny\rm df}}}{=}}

\newcommand{\state}{{\mathit state}}
\newcommand{\ttime}{{\mathit time}}
\newcommand{\pre}[1]{{\mathit {}^\bullet{#1}}}
\newcommand{\post}[1]{{\mathit {#1}^\bullet}}
\newcommand{\nextt}{{\sf succ}}
\newcommand{\word}{{\sf word}}
\newcommand{\true}{{\sf true}}
\newcommand{\false}{{\sf false}}
\newcommand{\init}{{\mathsf{init}}}
\newcommand{\Init}{{\mathsf{Init}}}

\newcommand{\fleche}{{\mathit \;{-}{-}{>}\;}}

\newcommand{\drop}[1]{}
\newcommand{\leer}{\varepsilon} 
\newcommand{\comm}[1]{\textcolor{red}{{#1}}}
\newcommand{\impl}{\Longrightarrow}
\newcommand{\rimpl}{\Longleftrightarrow}
\newcommand{\val}{{\mathcal V}}

\newtheorem{constraint}{Constraint}

\begin{document}

\title{Dynamic exploration of 
multi-agent systems \\ with timed periodic tasks} 
\author{Johan Arcile\\
  IBISC, Univ Evry, Université Paris-Saclay, \\
  91025, Evry, France
\and Raymond Devillers\\
  ULB, Bruxelles, Belgium
\and Hanna Klaudel\\
  IBISC, Univ Evry, Université Paris-Saclay, \\
  91025, Evry, France }

\runninghead{J. Arcile, R. Devillers, H. Klaudel}{Dynamic exploration of MAS with timed periodic tasks} 
\maketitle

\begin{abstract} 
We formalise and study multi-agent timed models 
\model{}s ({\em Multi-Agent with timed Periodic Tasks}),
where each agent is associated to a regular timed schema 
upon which all possibles actions of the agent rely.
\model{}s allow for an accelerated semantics and a layered structure of the 
state space, so that it is possible to 
explore the latter dynamically and use heuristics
to greatly reduce the computation time needed to address reachability problems.

We apply \model{}s to explore state spaces of autonomous vehicles and compare
it with other approaches in terms of expressivity, abstraction level and computation time.
\end{abstract}


\section{Introduction}

In the context of modelling and validating communicating autonomous vehicles (\cav{}s), 
the framework \verifcar{}~\cite{verifcar} allows to study the behaviour and properties 
of systems composed of concurrent agents interacting in real time (expressed through real variables called clocks) and through shared variables. 
Each agent performs time restricted actions that impact the valuation of shared variables. 
The system is highly non-deterministic due to overlaps of timed intervals 
in which the actions of various agents can occur.

\verifcar{} is suitable for the exhaustive analysis of critical situations in terms of safety, efficiency and robustness with a specific focus on the impact of latencies, communication delays and failures on the behaviour of \cav{}s.
It features a parametric model of \cav{}s allowing to automatically adjust the size of the state space to suit the required level of abstraction.
This model is based on timed automata with an interleaving semantics~\cite{AlurDill90}, 
and is implemented with \uppaal{}~\cite{uppaal4.1}, a state of the art tool for real time systems with an efficient state space reduction for model checking.
However, it is limited in terms of expressivity and deals only with discrete values, which is not always
convenient and may lead to imprecise computations.

Various approaches \cite{OKelly201616AA,kong15,platzer09,quottrup04}, 
relying on formal methods, address the modelling and analysis of multi-agent 
systems in a context similar to ours.
In particular, bounded model checking approaches \cite{Clarke2001,Biere2003,Sorea2003} 
have been used for studying temporal logic properties. 
Standard and highly optimised model checking tools, like \uppaal~\cite{uppaal4.1,larsen:97,AlurDill90}, simplify a lot the process of studying the behaviour of such systems, but have some drawbacks. 
For instance, in addition to 
clocks, they usually only allow integer variables while rational ones would sometimes be more natural, leading to artificial discretisations. 
Next, they only check Boolean expressions while it may be essential to analyse numerical ones.
Finally, the Boolean expressions are restricted to a subset of the ones allowed by classical logical languages, in particular by excluding nested queries.

It turns out that state spaces in the applications studied with \verifcar{} are generally very large but take the form of a semantic directed acyclic graph (\acygra{}). Each agent also has syntactically the form of a  \acygra{} between clock resets. 
Our goal is to exploit these peculiarities 
to build a dedicated checking environment for reachability problems. 
Concretely, we want to explore the graph dynamically (\ie, checking temporal logic properties directly as we explore states) to avoid constructing the full state space, and therefore not to loose time and memory space storing and comparing all previously reached states. 
The objective is to be able to tune the verification algorithm with heuristics that will choose which path to explore in priority, which might significantly speed up the computation time if the searched state exists.
That implies that our algorithms should explore the graph depth-first, since width-first algorithms cannot explore paths freely and are restricted to fully explore all the states at some depth before exploring the next one. 

For systems featuring a high level of concurrency between actions, such as the \cav{} systems, 
most of the non-determinism results from possibly having several actions of different agents available from a given state, that can occur in different orders and which often lead eventually to the same state. 
This corresponds in the state space to what is sometimes called diamonds.
Width first exploration allows to compare states at a given depth and therefore remove duplicates, which is an efficient way to detect such diamonds. 
On the other hand, depth-first exploring such a state space with diamonds, 
leads to examine possibly several times the same states or paths, 
which is not efficient. 
In this context, our aim is to detect and merge identical states coming from diamonds
while continuing to explore the state space mainly depth-firstly. 
This diamond detection will consist in a width-first exploration in a certain layer of the
state space, each layer corresponding to some states at a given depth having common characteristics.
It turns out that such layers may be observed in the state space of \cav{} systems.
This allows for a depth-first exploration from layer to layer, while greatly reducing the chances of exploring several times the same states.
The class of models on which this kind of algorithm 
can be applied will be referred to as {\em Multi-Agent with timed Periodic Tasks} (\model{}s).

To implement such algorithms we use \zinc{}~\cite{pommereau:hal-01941485}, a compiler for high level Petri nets that generates a library of functions allowing to easily explore the state space.
We use such functions to dynamically explore the state space with algorithms designed for our needs.
In particular, this allows to apply heuristics leading to faster computation times, and results in a better expressivity of temporal logic than \uppaal{}, in particular by including nested queries.
Another gain when comparing to \uppaal{} is that we are not limited to integer computations and can use real or rational 
variables, thus avoiding loss of information.
To use \zinc{}, we have to emulate the real time with discrete variables.
We do so in a way that preserves the behaviour of the system: when using the model with the same discrete variables as with \uppaal{}, we obtain identical results.

In this paper, we start by a formal definition of \gmodel{}s, a general class of \model-like models, 
study its properties and provide a translation for high level Petri nets.
Then, we introduce our \model{} models, by slightly constraining \gmodel{} ones, in order to avoid useless features and to allow a first kind of acceleration procedures. 
Then, we present the layered structure and the algorithms taking advantage of it.
Finally, we propose heuristics and use them in experiments that highlight the benefits of our approach in terms of expressivity, abstraction level and computation time.

\section{Syntax and semantics of \gmodel{}s}

A \gmodel{} is a model composed of several agents that may interact through a 
shared variable.
Each agent is associated with a clock and performs actions occurring in some given time intervals.
There is no competition between agents in the sense that no agent will ever have to wait for another one's action in order to perform its own actions.
However, there may be non determinism when several actions are available at the same time,
as well as choices between actions and time passing.

A \gmodel{} is a tuple $(\val,F,A,Init)$ where: 
\begin{itemize}
    \item $\val$ is a set of values;
    \item $F$ is a (finite) 
    set of variable transformations, \ie, calculable 
    functions from $\val$ to $\val$;
    \item $A$ is a set of $n$ agents such that $\forall i \in [1,n]$, agent $A_i \defeq (L_i,C_i,T_i,E_i)$ with:
    \begin{itemize}
        \item $L_i$ is a set of localities denoted as a list $L_i \defeq (l_i^1, \ldots, l_i^{m_i})$ with ${m_i} > 0$, such that $\forall i \neq j$, $L_i \cap L_j=\emptyset$;
        \item $C_i$ is the unique clock of agent $A_i$, with values in $\mathbb{N}$; 
        \item $T_i$ is a finite set of transitions, forming a directed 
        acyclic graph between localities, 
        with a unique initial locality $l_i^1$ and a unique final locality $l_i^{m_i}$, each transition being of the form $(l,f,I,l')$ where $l, l' \in L_i$ are the source and destination localities, 
        $f \in F$ is a function and $I \defeq [a,b]$ is an interval with $a,b \in \mathbb{N}$ and $a \leq b$.
        \item $E_i \in \mathbb{N}\setminus\{0\}$
        is the reset period of agent $A_i$.
    \end{itemize}
    \item $\Init$ is a triple $((l_1,\cdots,l_n),(\init_1,\cdots,\init_n),\init_V)$ where $\forall i \in [1,n]$, $l_i\in L_i$, 
    $\init_i \in \mathbb{N}$ 
    and $\init_V \in \val$.
\end{itemize}
For each agent $A_i$ and each locality $l\in L_i$,
we shall define by $\post{l}=\{(l,f,I,l')\in T_i\}$ the set of transitions originated from $l$, and by
$\pre{l}=\{(l',f,I,l)\in T_i\}$ the set of transitions leading to $l$. Note that, from the hypotheses, $\pre{l_i^1}=\emptyset$ and $\post{(l_i^{m_i})}=\emptyset$.
Moreover, when $i\neq j$, 
since $L_i\cap L_j=\emptyset$, $T_i\cap T_j=\emptyset$ too, so that each transition belongs to a single agent, avoiding confusions in the model.

A simple example of a \gmodel{} $M$ with two non-deterministic agents is represented in
Ex \ref{ex:mapt}.

In the semantics, for each agent $A_i$, 
we will emulate a transition from $l_i^{m_i}$ to $l_i^1$ that resets clock $C_i$ every $E_i$ time units.
As such, each agent in the network cycles over a fixed period. 
There can be several possible cycles though, since a given locality may be the source of several transitions, so that there may be several paths from $l_i^1$ to $l_i^{m_i}$. 

The behaviour of the system is defined as a transition system where $\Init$ is the initial state. 
A state of a \gmodel{} composed of $n$ agents as described above 
is a tuple denoted by $s = (\vec{l},\vec{c}, v)$ where $\vec{l} = (l_1,\cdots,l_n)$ with $l_i \in L_i$ is the current locality of agent $A_i$, $\vec{c} = (c_1,\cdots,c_n)$ where $c_i\in \mathbb{N}$ is the value of clock $C_i$, 
and $v \in \val$ is the value of variable $V$. 
There are three possible kinds of state changes: 
a firing of a transition, a clock reset and a time increase.
\begin{itemize}
    \item A transition $(l,f,[a,b],l') \in T_i$ can be fired if $l_i = l$ and $a \leq c_i \leq b$.
    Then, in the new state, 
    $l_i \leftarrow l'$ and $v \leftarrow f(v)$.
    \item A clock $C_i$ can be reset if $l_i = l_i^{m_i}$ and $c_i = E_i$. Then, 
    $c_i \leftarrow 0$ and $l_i \leftarrow  l_i^1$.
    \item Time can increase if 
    $\forall i\in[1,n]$, either there exist at least one transition $(l,f,[a,b],l') \in T_i$ with $l_i = l$ and $c_i < b$, or $l_i = l_i^{m_i}$ and $c_i < E_i$.
    A time increase means that $\forall i \in [1,n]$, $c_i \leftarrow c_i + 1$.
\end{itemize}

 It may be observed that there is a single global element in such a system: variable $V$; all the other ones are local to an agent.
 It is unique but its values may have the form of a vector, and an agent may modify several components of this vector through the functions of $F$ used in its transitions, 
 thus emulating the presence of several global variables.
 The values of $V$ are not restricted to the integer domain, but there is only a countable set of values that may be reached: 
 the ones that may be obtained from $\init_v$ by a recursive application of functions from $F$ 
 (the variable is not modified by the resets nor the time increases).
 However this domain may be dense inside the reals, for instance.
 
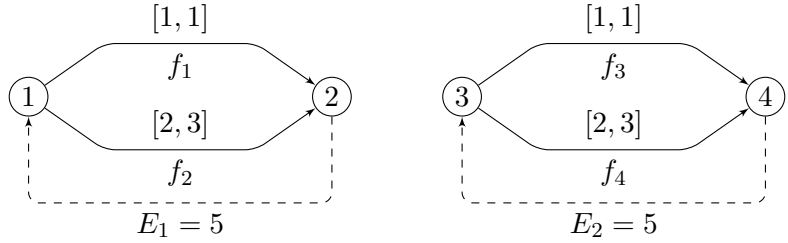
\begin{figure}[htb]
\tikzstyle{location}=[draw,circle,inner sep=2pt]
\centering
\begin{tikzpicture}[>=latex', xscale=1, yscale=.7,every node/.style={scale=1}]
\node[location] at (0,0) (a) {$1$};
\node[location] at (4,0) (b) {$2$};
\draw[->,rounded corners] (a) -- (1,1) -- node[above,midway] {$[1,1]$} node[below,midway] {$f_1$} (3,1) -- (b);
\draw[->,rounded corners] (a) -- (1,-1) -- node[above,midway] {$[2,3]$} node[below,midway] {$f_2$} (3,-1) -- (b);
\draw[->,dashed,rounded corners] (b) -- (4,-2) -- node[below,midway] {$E_1 = 5$} (0,-2) -- (a);
\end{tikzpicture}
\hspace{1cm}
\begin{tikzpicture}[>=latex',xscale=1, yscale=.7,every node/.style={scale=1}]
\node[location] at (0,0) (a) {$3$};
\node[location] at (4,0) (b) {$4$};
\draw[->,rounded corners] (a) -- (1,1) -- node[above,midway] {$[1,1]$} node[below,midway] {$f_3$} (3,1) -- (b);
\draw[->,rounded corners] (a) -- (1,-1) -- node[above,midway] {$[2,3]$} node[below,midway] {$f_4$} (3,-1) -- (b);
\draw[->,dashed,rounded corners] (b) -- (4,-2) -- node[below,midway] {$E_2 = 5$} (0,-2) -- (a);
\end{tikzpicture}
 \caption{\label{fig:exampleM} Visual representation of \gmodel{} from Ex. \ref{ex:mapt}. Dashed arcs represent resets.}
\end{figure}

\begin{example}
Let $M = (\val,F,A,\Init)$ where: 
\begin{itemize}
    \item $\val = \mathbb{R} \times \mathbb{N}$; 
    \item $F = \{f_1, f_2 ,f_3, f_4\}$ with\\
    $\begin{array}{ll}
         f_1(x,y) \rightarrow{} (2x,y+1) &
         f_2(x,y) \rightarrow{} (x+1.3,y+1) \\
         f_3(x,y) \rightarrow{} (\frac{x}{2},y) &
         f_4(x,y) \rightarrow{} (2x,y) \\
    \end{array}$
    \item $A = \{(L_1,C_1,T_1,E_1),(L_2,C_2,T_2,E_2)\}$ with \\
    $\begin{array}{lll}
         L_1 = \{1, 2\} & 
         T_1 = \{(1,f_1,[1,2],2),(l,f_2,[3,3],2)\} &
         E_1 = 5 \\
         L_2 = \{3,4\} &
         T_2 = \{(3,f_3,[1,2],4),(3,f_4,[3,3],4)\} &
         E_2 = 5 \\
    \end{array}$
    \item $\Init = ((1, 3),(0,0),(0.5,0))$.
\end{itemize}

A visual representation of $M$ is given in Fig.~\ref{fig:exampleM} while 
the initial fragment of its dynamics is depicted on top left of Fig.~\ref{fig:accelerated-space}.
Note that only transition firings and time increases are represented in Fig.~\ref{fig:accelerated-space} while the values $v$ of variable $V$ in the states $(\vec{l},\vec{c}, v)$ are always omitted. \hfill $\Diamond$

\label{ex:mapt}
\end{example}

In a dynamic system, {\em persistence} is a property that states that,
if two state changes are enabled at some state, then 
none of these changes disables the other one and performing them in any order leads to the same resulting state, forming a kind of {\em diamond}.
In \gmodel{} systems we have a kind of persistence restricted to different agents.
\begin{proposition}
In a \gmodel{}, if $i\neq j$ and $s=(\vec{l},\vec{c}, v)$ is any state, we have:
\begin{enumerate}
    \item if $s$ enables two transitions $t_1=(l_1,f_1,I_1,l'_1)\in T_i$ and $t_2=(l_2,f_2,I_2,l'_2)\in T_j$, 
    leading respectively to states $s_1$ and $s_2$,
    then $s_1$ enables $t_2$ leading to a state $s_3$ and $s_2$ enables $t_1$ leading to a state $s_4$;
    moreover $s_3=s_4$ iff $f_1\circ f_2(v)=f_2\circ f_1(v)$, i.e., if $f_1$ and $f_2$ commute on $v$;
    \item if $s$ enables a transition $t=(l,f,I,l')\in T_i$ and a reset of $A_j$,
    leading respectively to states $s_1$ and $s_2$,
    then $s_1$ enables the reset of $A_j$ leading to a state $s_3$ and $s_2$ enables $t$ leading to the same state $s_3$;
    \item if $s$ enables a reset of $A_i$ and a reset of $A_j$, 
    leading respectively to states $s_1$ and $s_2$,
    then $s_1$ enables the reset of $A_j$ leading to a state $s_3$ and $s_2$ enables the reset of $A_i$ leading to the same state $s_3$.
\end{enumerate}
\end{proposition}

\begin{proof} 
\begin{enumerate}
    \item The property results from the fact that transitions do not modify clocks, and a transition of some agent only modifies the locality of the latter, together with the common variable $V$;
    in $s_3$ the variable becomes $f_2\circ f_1(v)$, while in $s_4$ the variable becomes $f_1\circ f_2(v)$.
    Note that if $i=j$, $s_1$ does not enable transition $t_2$ since, from the acyclicity hypothesis, the locality of $A_i$ is changed, hence is not the source of $t_2$ (and symmetrically);
    \item the property results from the fact that a transition does not modify any clock, 
    and a reset of $A_j$ only modifies the locality and clock of the latter; 
    the common variable $V$ will have the value $f(v)$ after both the execution of the transition followed by the reset as well as after the reset followed by the transition. 
    Note that $i$ may not be the same as $j$ here since the reset of $A_i$ may only occur in locality $l_i^{m_i}$, while no transition may occur there;
    \item the property results from the fact that a reset of an agent only modifies the locality and clock of the latter.
    Note that, if $i=j$, since after a first reset the locality becomes $l_i^1$,
    so that a second one may only occur if $l_i^1=l_i^{m_i}$, \ie, $A_i$ has a single locality, no transition and does not act on the common variable, hence may be dropped.
\end{enumerate}
\end{proof}

\subsection{Constraints}

We define in this subsection two types of constraints, which are motivated by the properties of our target application domain and which will be used to obtain interesting properties.

A \gmodel{} is called \emph{strongly live} if it satisfies the following constraint:

\begin{constraint}  \hspace*{1cm}
 \begin{enumerate}
    \item If the initial locality of some agent $A_i$
    is the terminal one ($l_i=l_i^{m_i})$, then the initial value of clock $C_i$ satisfies $\init_i\leq E_i$;
    \item otherwise (when $l_i\neq l_i^{m_i}$), we have $\init_i\leq \max\{b \mid (l_i,f,[a,b],l')\in \post{l_i}\})$;
    \item moreover, for each agent $A_i$, if $l\in L_i\setminus\{l_i^1,l_i^{m_i}\}$, then \\
    $\max\{b \mid (l',f,[a,b],l)\in \pre{l}\} \leq \min\{b' \mid (l,f',[a',b'],l'')\in \post{l}\}$;
    \item and we have $\max\{b \mid (l',f,[a,b],l_i^{m_i})\in \pre{l_i^{m_i}}\} \leq E_i$.
\end{enumerate}   
\label{constr:liveness}
\end{constraint}

The first two constraints ensure that, when the system is started, either in the terminal locality or in a non terminal one of some agent, the time is not blocked and we shall have the possibility to perform an action in some future.
The next constraint ensures that whenever a non terminal locality is entered, 
any (and not only some) transition originated from that locality will have the possibility to occur in some future (the case when we enter an initial locality is irrelevant since resets reinitialise the corresponding clock to 0).
The last constraint captures similar features in the case when we enter a terminal locality.

In other words, each transition or reset in $\post{l}$ is enabled when entering $l$ or will be enabled in the future (after possibly some time passings in order to reach the lower bound $a$).

\begin{proposition} 
\label{prop:constraints1}
In a \gmodel{} satisfying Constraints~\ref{constr:liveness},
after any evolution $\omega$, each transition and each reset (as well as time passings) may be fired in some future.
\end{proposition}

\begin{proof}
We may first observe by induction on the length of $\omega$ that, if $s=(\vec{l},\vec{c},v)$ is the state reached after the evolution $\omega$, for any agent $A_i$ we have
$(l_i=l_i^{m_i})\impl c_i\leq E_i$ and $(l_i\neq l_i^{m_i})\impl c_i\leq \max\{b \mid (l_i,f,[a,b],l')\in \post{l_i}\})$,
\ie, the same \gmodel{} with initial state $s$ also satisfies Constraint~\ref{constr:liveness} (the last two ones do not rely on the initial state).
The property is trivial for $\omega=\leer$.
If the property is satisfied for some $\omega$, it remains so for any extension, from the definition of the semantics of \gmodel{} and the last two points of Constraint~\ref{constr:liveness}.

It also results from the same remark that any evolution $\omega$ satisfying the mentioned properties may be extended (there is no deadlock). 
It remains to show that any extension may be performed in some future.

For time passing, we may observe that, if time passing may never be performed, since the set of localities for each agent has the form of a \acygra, extending $\omega$ will finally perform a reset, and since each $E_i$ is strictly positive we shall finally perform all the resets and stop at some point, which is forbidden. Hence we are sure time passings will be possible.

Since time passings may always be performed in the future, all resets will be performed eventually.

Finally, let $t=(l'_i,f,[a,b],l''_i)\in T_i$. From the same argument about time passings, 
it will be possible to eventually perform reset $r_i$, 
then follow a path going from $l_i^1$ to $l'_i$ in the \acygra{} of $A_i$, 
performing each transition in turn when reaching the corresponding $a$, 
due to Constraint~\ref{constr:liveness}.3.
\end{proof}

The next constraint is a syntactic manner of ensuring the
acyclicity of the \gmodel's dynamics (\ie, its state space): 
\begin{constraint} \hspace*{1cm}
\begin{enumerate} 
    \item $\val \defeq W \times X$ and there exist an agent $A_i$ such that 
    in all paths of transition from $l_i^1$ to $l_i^{m_i}$, there exists a transition $t\defeq (l,f,[a,b],l')$ 
    such that for all $(w,x) \in \val$, $f(w,x) = (w',x')$ with $x < x'$;
    \item and there is no $f \in F$ such that for some $(w,x)\in \val$,
    we have $f(w,x) = (w',x')$ with $x > x'$.
\end{enumerate}
\label{constr:acyclic}
\end{constraint}

The first constraint ensures that at least one agent increments the $X$ part of variable $V$ at least once between two of its resets.
The second one ensures that no function can decrease the $X$ part of variable $V$.

In other words, the $X$ part of $v$ increases in each cycle of agent $A_i$, which results in an absence of cycles in the whole state space of the \gmodel.

\begin{proposition}\label{DAG.prop}
A \gmodel{} satisfying Constraint \ref{constr:acyclic} is acyclic.
\end{proposition}

\begin{proof}
If there is a cycle, it means that there exists a path from the initial state with at least two different states in the path $s_1=(\vec{l},\vec{c},v)$ and $s_2=(\vec{l'},\vec{c'},v')$, which are actually identical.
Since the localities of each agent form a static \acygra{} determined by its transitions and each $E_i$ is strictly positive, 
having $\vec{l}=\vec{l'}$ and $\vec{c}=\vec{c'}$ may only happen
if all agents have done at least one reset between $s_1$ and $s_2$. 
Indeed, if there is no reset, since $\vec{c}=\vec{c'}$ we may only have transitions in the cycle, 
but this is incompatible to have \acygra{}s in each agent.
Moreover, since between two resets of an agent time strictly increases, 
$\vec{c}=\vec{c'}$ may only occur if all agents have performed one or more resets.
Thus, it is enough to observe that variable $V$ cannot decrease from Constraint~\ref{constr:acyclic}.2, 
and a transition $t$ like in 
Constraint~\ref{constr:acyclic}.1 should occur, 
guaranteeing that
$v\neq v'$. 
\end{proof}

\begin{definition}
A \model{} 
is a \gmodel{} satisfying Constraints \ref{constr:liveness} and \ref{constr:acyclic}.
\end{definition}

A \model{} may be non-deterministic but it is strongly live and has a \acygra{} state space.
For instance, the \gmodel{} $M$ from Ex.~\ref{ex:mapt} satisfies both constraints 
(acyclicity is satisfied due to $y$ being incremented in all cycles of $A_1$) and so is actually a \model.

\section{Translation into high level Petri nets}

A high level Petri net \cite{DBLP:series/eatcs/Jensen92}
can be viewed as an abbreviation of a low-level one \cite{peterson} 
where tokens are elements of some set of values that can be checked and 
updated when transitions are fired.
Here, we express a \gmodel{} as a high level Petri net to be implemented with \zinc{}.

Formally, a high level Petri net is a tuple $(S,T,\lambda,M_0)$ where:
\begin{itemize}
    \item $S$ is a finite set of places;
    \item $T$ is a finite set of transitions;
    \item $\lambda$ is a labelling function on places, transitions and arcs such that
    \begin{itemize}
        \item for each place $s \in S$, $\lambda{}(s)$ is a set of values defining the type of $s$, 
        \item for each transition $t \in T$, $\lambda{}(t)$ is a Boolean expression with variables and constants defining the guard of $t$ and 
        \item for each arc $(x,y) \in (S\times{}T)\cup(T\times{}S)$, $\lambda{}(x,y)$ is the annotation of the arc from $x$ to $y$, driving the  
        production or consumption of tokens.
    \end{itemize}
    \item $M_0$ is an initial marking associating tokens to places, according to their types.
\end{itemize}

The semantics of a high level Petri net is captured by a transition system
containing as states all the markings, which are reachable from the initial marking $M_0$.
A marking $M'$ is directly reachable from a marking $M$ if there is a transition $t$ enabled at $M$, whose
firing leads to $M'$; it is reachable from $M$ if there is a sequence of such firings leading to it.
A transition $t$ is enabled at some marking $M$ if the tokens in all the input places of $t$ allow to satisfy
the flow expressed by the annotations of input arcs and the guard of $t$, through a valuation of the variables involved in the latter.
The firing of $t$ consumes the concerned tokens in input places of $t$ and produces 
tokens on output places of $t$, according to the annotations of the output arcs and the same valuation.

\begin{figure*}[htb]
\tikzstyle{location}=[draw,circle,inner sep=2pt]
\tikzstyle{transition}=[draw,rectangle,inner sep=3pt]
\centering
\begin{tikzpicture}[>=latex', xscale=1, yscale=.7,every node/.style={scale=1}]
\node[location, label=left:{$s_1$}] at (-3,0) (s1) {$1$};
\node[location, label=left:{$s_2$}] at (0,-2) (s2) {$2$};
\node[location, label=right:{$s_3$}] at (4,0) (s3) {$(0,0)$};
\node[transition, label ={$x>0$}] at (0,0) (t) {$t$};
\draw[->] (s1)  -- node[above,midway] {$x$} (t);
\draw[->] (s2)  -- node[left,midway] {$y$} (t);
\draw[->] (s3)  -- node[above,midway] {$(w,z)$} (t);
\draw[->, rounded corners] (t)  -- (1,-1) --node[above,midway] {$(w+x,z+y)$} (3,-1)-- (s3);
\end{tikzpicture}
 \caption{\label{fig:petri_firing} A high level Petri net. } 
\end{figure*}
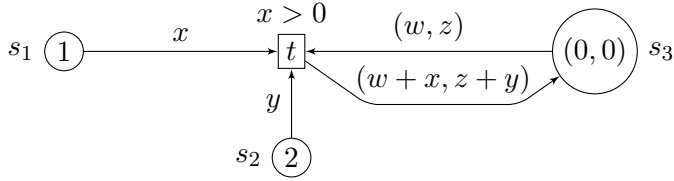

Fig.~\ref{fig:petri_firing} shows an example of a high level Petri net where
place types are $\mathbb N$ for $s_1$ and $s_2$, and $\mathbb N\times \mathbb N$ for $s_3$.
Transition $t$ is enabled at the initial marking since there exists a valuation of variables
in the annotations of arcs and in the guard of $t$, with values from tokens,
$x\mapsto 1$, $y\mapsto 2$, $w\mapsto 0$, $z\mapsto 0$, 
that satisfies the guard. The firing of $t$ consumes the tokens in all three places
and produces a new token $(1,2)$ in place $s_3$.

\begin{definition}
\label{def:translation}
Given a \gmodel{} $Q = (\val,F,A,Init)$ with $|A| = n$, its translation to a high level Petri net $N = translate(P) = (S,T,\lambda,M_0)$ is defined as follows:
\begin{itemize}
    \item $S = \{s_A,s_C,s_V\}$ with $\lambda(s_A) \defeq L_1 \times \cdots \times L_n$ where $L_i$ is the set of localities of agent $A_i$ (its $j$th element will be denoted $l_i^j$), 
    $\lambda(s_C) \defeq  \mathbb N^n$ and $\lambda(s_V) \defeq \val$; 
    For any token $x$ of the type $\lambda(s_A)$ or $\lambda(s_C)$, we denote by $x[i]$ the $i^{\textrm{\tiny{th}}}$ element of the list.
    \item $T \defeq  T_{trans} \cup T_{reset} \cup \{ t_{time} \}$ where
    \begin{itemize}
        \item $T_{trans}$ is the smallest set of transitions such that, for each agent $A_i=(L_i,C_i,T_i,E_i)$ in $A$ 
        and for each transition $(l,f,[a,b],l') \in T_i$, there is a transition $t \in T_{trans}$ such that $\lambda(s_A,t) \defeq  x$, $\lambda{}(s_C,t) \defeq  y$, $\lambda{}(s_V,t) \defeq  z$, $\lambda(t,s_A) \defeq  x'$ where $x'[i] \leftarrow l'$ and $\forall j \neq i$, $x'[j] \leftarrow x[j]$, $\lambda{}(t,s_C) \defeq  y$, $\lambda(t,s_V) \defeq  f(z)$ and $\lambda{}(t) \defeq (x[i] = l) \land (a \leq y[i] \leq b)$. This is equivalent to the set of transitions of the \gmodel{}.
        \item $T_{reset}$ is the smallest set of transitions such that, for each agent $A_i=(L_i,C_i,T_i,E_i)$ in $A$, there is a transition $t\in T_{reset}$ such as $\lambda(s_A,t) \defeq x$, $\lambda{}(s_C,t) \defeq  y$, $\lambda(t,s_A) \defeq  x'$ where $x'[i] \leftarrow l_i^1$ and $\forall j \neq i$, $x'[j] \leftarrow x[j]$, $\lambda{}(t,s_C) \defeq y'$ where $y'[i] \leftarrow 0$ and $\forall j \neq i$, $y'[j] \leftarrow y[j]$, and $\lambda{}(t) \defeq (x[i] = l_i^{m_i}) \land (y[i] =  E_i)$ where ${m_i} \defeq |L_i|$. This is equivalent to the set of clock resets of the \gmodel{}.
        \item $\lambda(s_A,t_{\ttime}) \defeq x$, $\lambda{}(s_C,t_{\ttime}) \defeq y$, $\lambda(t_{\ttime},s_A) \defeq x$, $\lambda{}(t_{\ttime},s_C) \defeq y'$, where $\forall i \in [1,n]$, $y'[i] \leftarrow y[i] + 1$, and $\lambda{}(t_{\ttime}) \defeq G_1 \land \cdots \land G_n$ where $G_i$ acts as the "upper bound guard" for all the transitions in agent $A_i$, i.e., 
        $G_i \defeq (g_1 \lor \cdots \lor g_{m_i})$ with ${m_i} = |L_i|$ and $\forall j \in [1,{m_i}-1]$, $g_j \defeq (x[i] = l_i^j) \land (y[i] < B$), where $B = \max\{b|(l_i^j,f,[a,b],l')\in \post{l_i^j}\}$ is the highest upper bound of the intervals from all outgoing transitions of $l_i^j$ 
        and $g_{m_i} \defeq (x[i] = l_i^{m_i}) \land (y[i] < E_i)$. This is equivalent to a time increase.
    \end{itemize}
    \item $(M_0(s_A),M_0(s_C),M_0(s_V)) = \Init$ is the initial marking 
\end{itemize}
\end{definition}

The translation associates singletons as arc annotations for all arcs.
As a consequence, during the execution, starting from the initial marking which associates one token to each place, 
there will always be exactly one token in each of the three places.
Each reachable marking, where $s_A$ contains $\vec{l}$, $s_C$ contains $\vec{c}$ and $s_V$ contains $v$,
encodes a state $(\vec{l},\vec{c}, v)$ of the considered \gmodel.

Figure~\ref{fig:petriM} sketches the Petri net translation of the \gmodel{} from Ex.~\ref{ex:mapt}.
At the initial marking, $t_1$, $t_2$, $t'_1$ and $t'_2$ are not enabled because 
the token read from $s_C$ (\ie, the vector of clock values) does not satisfies the guards, 
while $r_1$ and $r_2$ are not enabled because the token read from $s_A$ (\ie, the vector of localities) 
does not satisfies the guards.
On the other hand, the transition $\ttime$ is enabled.
Its firing reads\footnote{means that consumes and produces the same tokens} tokens in places $s_V$ and $s_A$, consumes $(0,0)$ and produces $(1,1)$ in $s_C$.
At this new marking, $\ttime$, $t1$ and $t2$ are enabled and the process continues 
exactly as in the \gmodel{}.

\begin{figure*}[htb]
\tikzstyle{location}=[draw,circle,inner sep=2pt]
\tikzstyle{transition}=[draw,rectangle,inner sep=3pt]
\centering
\begin{tikzpicture}[>=latex', xscale=1, yscale=.7,every node/.style={scale=1}]
\node[location, label=right:{$s_A$}] at (0,0) (sa) {$(1,3)$};
\node[location, label=left:{$s_C$}] at (-10,0) (sc) {$(0,0)$};
\node[location, label=below:{$s_V$}] at (-5,-8) (sv) {$(0.5,0)$};
\node[transition, label ={$\begin{array}{l} ((l_1 = 1 \land c_1 < 3) \lor (l_1 = 2 \land c_1 < 5)) \\ \land ((l_2 = 3 \land c_2 < 3) \lor (l_2 = 4 \land c_2 < 5))
\end{array}$}] at (-5,4) (t) {$\ttime$};
\node[transition, label={$l_1 = 2 \land c_1 = 5$}] at (-5,2) (r1) {$r_1$};
\node[transition, label={$l_2 = 4 \land c_2 = 5$}] at (-5,0) (r2) {$r_2$};
\node[transition, label=below right:{$l_1 = 1 \land 1 \leq c_1 \leq 1$}] at (-10,-8) (f1) {$t_1$};
\node[transition, label=below:{$l_2 = 3 \land 1 \leq c_2 \leq 1$}] at (-8.5,-4.5) (f3) {$t_2$};
\node[transition, label=below left:{$l_1 = 1 \land 2 \leq c_1 \leq 3$}] at (0,-8) (f2) {$t'_1$};
\node[transition, label=below:{$l_2 = 3 \land 2 \leq c_2 \leq 3$}] at (-1.5,-4.5) (f4) {$t'_2$};
\draw[<->,rounded corners] (sa)  |- node[below,near end] {$(l_1,l_2)$} (t);
\draw[<->,rounded corners] (sa)  -- node[above,midway] {$(l_1,l_2)/(1,l_2)$} (r1);
\draw[<->,rounded corners] (sa)  -- node[above,midway] {$(l_1,l_2)/(l_1,3)$} (r2);
\draw[<->,rounded corners] (sa)  -- node[below,near start] {$(l_1,l_2)/(2,l_2)$} (f1);
\draw[<->,rounded corners] (sa)  -- node[above,near start] {$(l_1,l_2)/(l_1,4)$} (f3);
\draw[<->,rounded corners] (sa)  -- node[above,midway] {$(l_1,l_2)/(2,l_2)$} (f2);
\draw[<->,rounded corners] (sa)  -- node[above,midway] {$(l_1,l_2)/(l_1,4)$} (f4);
\draw[<->,rounded corners] (sc)  |- node[below,near end] {$(c_1,c_2)/(c_1+1,c_2+2)$}  (t);
\draw[<->,rounded corners] (sc)  -- node[above,midway] {$(c_1,c_2)$} (r1);
\draw[<->,rounded corners] (sc)  -- node[above,midway] {$(c_1,c_2)$} (r2);
\draw[<->,rounded corners] (sc)  -- node[above,midway] {$(c_1,c_2)$} (f1);
\draw[<->,rounded corners] (sc)  -- node[above,midway] {$(c_1,c_2)$} (f3);
\draw[<->,rounded corners] (sc)  -- node[above,near start] {$(c_1,c_2)$} (f2);
\draw[<->,rounded corners] (sc)  -- node[above,near start] {$(c_1,c_2)$} (f4);
\draw[<->,rounded corners] (sv)  -- node[above,midway] {$(x,y)/(2x,y+1)$} (f1);
\draw[<->,rounded corners] (sv)  -- (-6,-5) -- node[above,near start] {$(x,y)/(\frac{x}{2},y)$} (f3);
\draw[<->,rounded corners] (sv)  -- node[above,midway] {$(x,y)/(x+1.3,y+1)$} (f2);
\draw[<->,rounded corners] (sv)  -- (-4,-5) -- node[above,near start] {$(x,y)/(2x,y)$} (f4);
\end{tikzpicture}
 \caption{\label{fig:petriM} Petri net translation of the \gmodel{} from Ex.~\ref{ex:mapt} with the initial marking. Arcs are bidirectional and annotated by pairs $w/z$ (or $w$ instead of $w/w$) meaning that $w$ is the label of the arc from place to transition and $z$ of the opposite one.} 
\end{figure*}
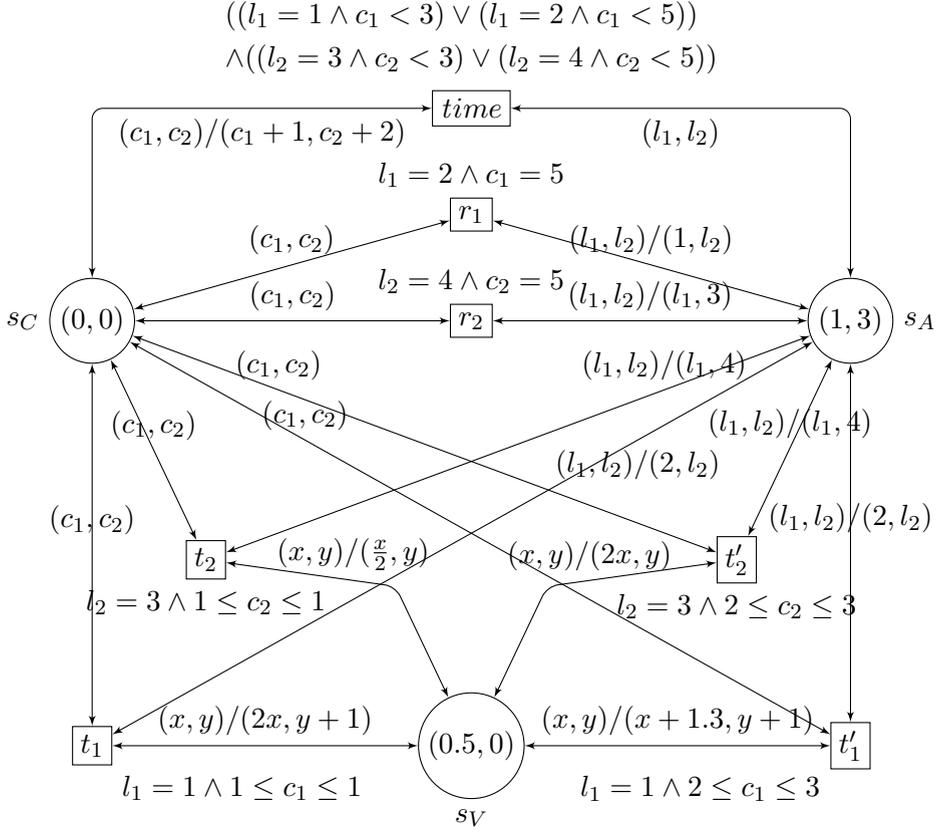

\begin{proposition}
A \gmodel{} $Q$ and its translated Petri net $N = translate(Q)$ have equivalent state spaces and semantics.
\end{proposition}

\begin{proof}
Immediate from the definitions.
\end{proof}

\FloatBarrier

\section{Acceleration} 

Let us assume we are interested by the causality relation between transitions rather than by the exact dates of their firings. It is then often possible to reduce the size of the original state space.
To do so, we assume that the \gmodel{} satisfies Constraint \ref{constr:liveness}. 

We may first consider {\em action zones},
defined as maximum time intervals in which the same transitions and resets are enabled from the current state 
(note however that, when a reset is enabled, the zone has length $0$, 
since before an $E_i$ the corresponding reset is not enabled and we may not go beyond $E_i$). 
Instead of increasing time by unitary steps, we can then progress in one step 
from an action zone to another one, until we decide or must fire a transition or reset. 
This generally corresponds to increasing time by more than one unit at once. 
Note that, when we jump to a zone, we may choose any point in it since, by definition,
all of them behave the same with respect to enabled events;
in the following, we have chosen to go to the end of the zone since this allows to perform bigger time steps; this will also be precious when defining borders of layers. 

However, not all action zones reachable from a given state need to be explored: we may neglect action zones which are {\em dominated} by other ones, \ie, for which the set of enabled transitions and resets is included in another one.
As an example, let us assume that the sets of enabled transitions are successively (from the current state):
$\{t1,t2\}\rightarrow{\bf\{t1,t2,t3\}}\rightarrow\{t2,t3\}\rightarrow\{t3\}\rightarrow\{t3,t4\}\rightarrow{\bf\{t3,t4,t5\}}\rightarrow\{t4,t5\}\ldots$
(i.e., letting the time evolve, we first encounter $a3$, then $b1$, $b2$, $a4$, $a5$, $b3$, ...).
The maximal (non-dominated) action zones are indicated in bold.
We may thus first jump (in time) to the end of 
the zone allowing $\{t1,t2,t3\}$, then choose if we want to fire $t1$ or $t2$ or $t3$, or jump to the end of 
the zone allowing $\{t3,t4,t5\}$, where we can again decide to fire a transition or not (unless a firing is mandatory, \ie, there is no further non-dominated action zone). 

In order to pursue the analysis, let $s=(\vec{l},\vec{c},v)$ be the current state and, 
for each agent $A_i$ let 
\[
\begin{array}{lll}
B_i &\defeq  & 
    \left\{ \begin{array}{l}
        E_i-c_i  \mbox{ if } l_i=l_i^{m_i} \\
        \max\{b-c_i \mid {(l_i,f,[a,b],l')\in \post{l_i}}\}  \mbox{ otherwise} \\
        \end{array}\right. \\ \\
B & \defeq & \min\{B_i \mid i\in[1,n]\}
\end{array}
\]

From our hypotheses, each $B_i$, hence also $B$, is non-negative, and we may not let pass more than $B$ time units before choosing to fire a transition or a reset.
In particular, if $B=0$, increasing time would prevent any transition in some locality to ever be enable again.
To avoid such a local deadlock, we must choose a transition or reset to fire.

When time evolves, if we reach an $a$ or an $E$ the set of enabled transitions and resets increases 
(note that an $E$ behaves both as an $a$ and as a $b$), 
and if we overtake a $b$ (we may not overtake an $E$), this set shrinks.
We must thus find the first $b$ or $E$ preceded by at least one $a$.

This may be done as follows: let
\[
\begin{array}{lll}
{\mathbf a} &\defeq  & 
    \left\{ \begin{array}{ll}
        \min(\alpha) & \mbox{if } \alpha \defeq \{ a-c_i \mid \exists\; (l_i,f,[a,b],l'_i) \in T_i \mbox{ for some agent $A_i$} \\
        & \mbox{and } c_i<a\leq B \}~ \cup \{E_i-c_i \mid l_i = l_i^{m_i} \mbox{ for some agent $A_i$} \\
        & \mbox{and } c_i<E_i\leq B \} \neq \emptyset \\
        0  & \mbox{otherwise} \\
        \end{array}\right. \\ \\
\delta &\defeq &     
    \left\{ \begin{array}{ll}
            \min(\beta) & \mbox{if } \beta \defeq \{ b-c_i \mid \exists\;(l_i,f',[a,b],l')\in \post{l_i} \mbox{ for some agent $A_i$} \\ 
                & \mbox{with $0 < \mathbf{a}\leq b-c_i\leq B$}\} \cup \{E_i-c_i \mid l_i = l_i^{m_i} \mbox{ for some } \\ 
                & \mbox{agent $A_i$ with $E_i\leq c_i+B$ }\} \neq \emptyset \\
            0 & \mbox{otherwise} \\
        \end{array}\right. \\

\end{array}
\]

It may be observed that 
${\bf a}>0\rimpl\alpha\neq\emptyset\rimpl \beta\neq\emptyset\rimpl \delta>0$.

If ${\mathbf a} = 0$, that means that there is no way to increase the set of enabled transitions or resets in the future; there is thus no interest to let time evolve (and indeed $\delta=0$) and we must choose now a transition or a reset to be fired (time will possibly be allowed to increase in the new locality). 
Otherwise, we may fire a transition or a reset or perform a time jump of $\delta$, which is of at most $B$ time units.
Note that when a transition or reset is fired,
we need to recompute $B$; when a time shift (or jump) is performed, $\delta>0$ and we need to adjust all the clocks and $B$: $\forall i:c_i\leftarrow c_i+\delta$ and $B\leftarrow B-\delta$.

We may remark that the initial state as well as the states reached after a firing are not necessarily in a maximal zone. This may be checked easily: the first $b$ or $E$ is preceded by one or more $a$'s from this current state.
We may then force a time jump $\delta$ as computed above to reach the end of the first maximal zone before wondering if we shall perform a firing. 
However, this is not absolutely necessary and we may decide to perform a firing or a time jump at this current state. 

Finally, we may observe that, when we perform a firing in some agent $A_i$, we are positioned before $B$, hence before $B_i$ by definition.
From Constraints \ref{constr:liveness}.3 and \ref{constr:liveness}.4 
above, whatever the time jumps performed in the previous state, the first maximal zone is the same,
so that we do not miss a possible firing from the new current state. 

\begin{figure}[htb]
\begin{center}
\includegraphics[scale=0.4]{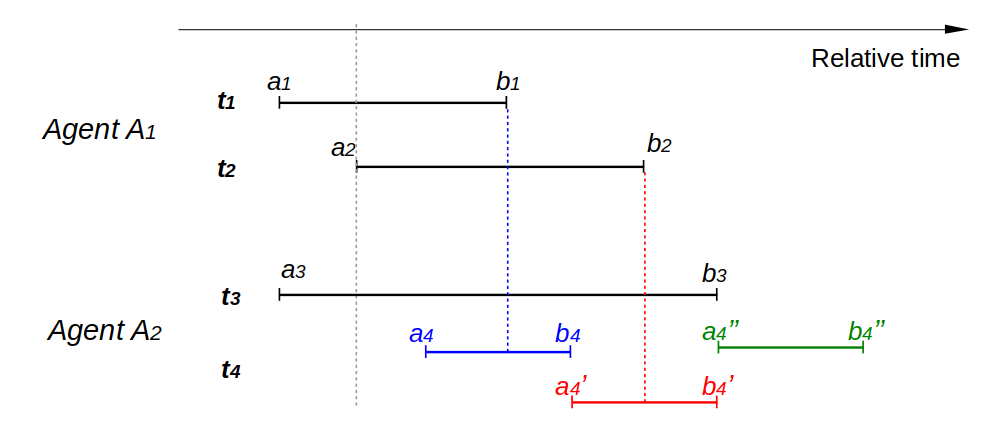} 
\caption{\label{fig:acceleration}Example of time increase based on the action zone acceleration. Current time is indicated by grey dots, while the maximal possible time increase for each variant is shown with its respective color (blue for $t_4$, red for $t_4'$ and green for $t_4''$).} 
\end{center}
\end{figure}

For a better understanding, let us consider the example of acceleration illustrated in Fig.~\ref{fig:acceleration},
with two agents $A_1$ and $A_2$.
To simplify the presentation, we shall assume that the clocks $C_1$ and $C_2$ are aligned, so that 
the time intervals of the transitions can be represented in the same space.
The current state of the system can be described as follows:
from the current locality of agent $A_1$, the outgoing transitions $t_1$ and $t_2$ 
can be fired respectively in the intervals $[0,3]$ and $[1,5]$,
while from the current locality of agent $A_2$, the outgoing transitions $t_3$ and $t_4$ 
(in blue in the figure) can be fired respectively in $[0,6]$ and $[2,4]$.
For any transition $t_i$, its lower and upper bounds will be referred to as $a_i$ and $b_i$.
Let us assume that the current time is currently at instant $1$. 
In such a case, the current action zone is $[a_2,a_4[$ and it enables $t_1$, $t_2$ and $t_3$.
The next action zone $[a_4,b_1]$ enables all four transitions (and is thus maximal). 
So, from the current action zone we may fire one of $t_1$, $t_2$ or $t_3$ or let the time pass.
The (accelerated) time increase should lead then to action zone $[a_4,b_1]$, for instance at $b_1$.
That way, we would include all the possible sequences of transitions, 
including the firing of $t_4$ followed by $t_1$. 

Now let us consider a variant of the example, in which $t_4$ is replaced by $t_4'$ 
(in red in the figure), with a time interval of $[4,6]$.
In that scenario, the current action zone is $[a_2,b_1]$ and it enables $t_1$, $t_2$ and $t_3$. 
The next zone is $]b1,a4'[$, which enables $t2$ and $t3$ only:
since the enabled transitions are included in the current action zone, this zone is not interesting
from a causality point of view. 
Finally, the zone $[a_4',b_2]$ enables $t_2$, $t_3$ and $t_4$, 
which is interesting because a new transition becomes enabled and 
the next time increase should lead to the end of this zone.
It is important to note that all the sequences involving $t_1$ are preserved, as the time increase 
is only one of the possibles evolution of the system, the firing of $t_1$, $t_2$ and $t_3$ also being possible. 

Finally, let us consider a second variant, in which $t_4$ is replaced by $t_4''$ (in green in the figure), 
with a time interval of $[6,8]$.
In that scenario, the current action zone is still $[a_2,b_1]$, which enables $t_1$, $t_2$ and $t_3$. 
The next action zone $]b1,b2]$ enables $t2$ and $t3$ only.
As before, the enabled transitions are included in the current action zone, 
which means that going to this zone is irrelevant. 
However, it is not possible to go further ahead since we reached $B$ 
(it corresponds to the time before reaching $b_2$). 
We must thus chose a transition to fire in the current zone. 
Transition $t_4''$ is presently 
non-enabled; it may become enabled in the future however,
after (at least) $t_1$ or $t_2$ is fired.

In the context of Petri net the acceleration may be defined syntactically (by modifying the guard of transition $t_{time}$ 
and the annotations of arcs from transition $t_{time}$ to place $s_C$)
and corresponds to replacing the following items in Definition \ref{def:translation}:
\begin{itemize}
    \item $\lambda{}(t_{time},s_C) = y'$, where $\forall i \in [1,n]$, $y'[i] \leftarrow y[i] + \delta$;
    \item $\lambda{}(t_{time}) \defeq \delta > 0$.
\end{itemize}
The computation of $\delta$ is possible thanks to the current localities of agents present 
in tokens in place $s_A$ and the values of clocks present as tokens in place $s_C$.

Note finally an interesting feature of the accelerated semantics: 
if we change the granularity of the time and multiply all the timing constant by some factor, the size of the state space of the original semantics is inflated accordingly.
On the contrary, the size (and structure) of the accelerated semantics remains the same.

\subsection{Abstracted dynamics} 

In order to capture the causality feature of such models, 
mixing time passings and transition/reset executions,
and to drop the purely timed aspects,
we shall consider the graph whose nodes are the projections of evolutions from the initial state on the set of transitions and resets. 
Said differently, if we have a word on the alphabet composed of $+\delta$ (time passing, with $\delta=1$ in the original, non-accelerated, semantics), $t_{i,j}$'s (transitions of agent $A_i$) and $r_i$'s (reset of agent $A_i)$ representing a possible evolution of the system up to some point,
by dropping all the $+\delta$'s we shall get its projection, and a node of the abstracted (from timing aspects) graph.
The (labelled) arcs between those nodes will be defined by the following rule: if $\alpha$ and $\alpha t$ (or $\alpha r$) are two nodes, there is an arc labelled $t$ (or $r$) between them.
This will define a (usually infinite) labelled tree, abstracted unfolding of the semantics (either original or accelerated) of the considered system.

The initial node (corresponding to the empty evolution) will be labelled by the projection $(\vec{l};v)$ of the initial state $(\vec{l},\vec{c};v)$.
This will automatically (recursively) determine the label of the other nodes: if $(\vec{l};v)$ is the label of some node and there is an arc labelled $t=(l_i,f,[a,b],l'_i)$ from it to another one, the latter will be labelled $(\vec{l'};f(v))$, where $\vec{l'}$ is $\vec{l}$ with $l_i$ replaced by $l'_i$; 
and if the arc is labelled $r_i$, the label of the destination node will be $(\vec{l'};v)$, where $\vec{l'}$ is $\vec{l}$ with $l_i$ replaced by $l_i^1$.

As an illustration consider the \model{} of Ex. \ref{ex:mapt}, where we neglect the values of the variable to simplify a bit the presentation. 
The initial fragment of the original and accelerated dynamics as well as the corresponding abstracted dynamics are represented in Fig.~\ref{fig:accelerated-space},
assuming initially the clocks are both equal to $0$, $A_1$ is in state $1$ and $A_2$ is in state $3$.

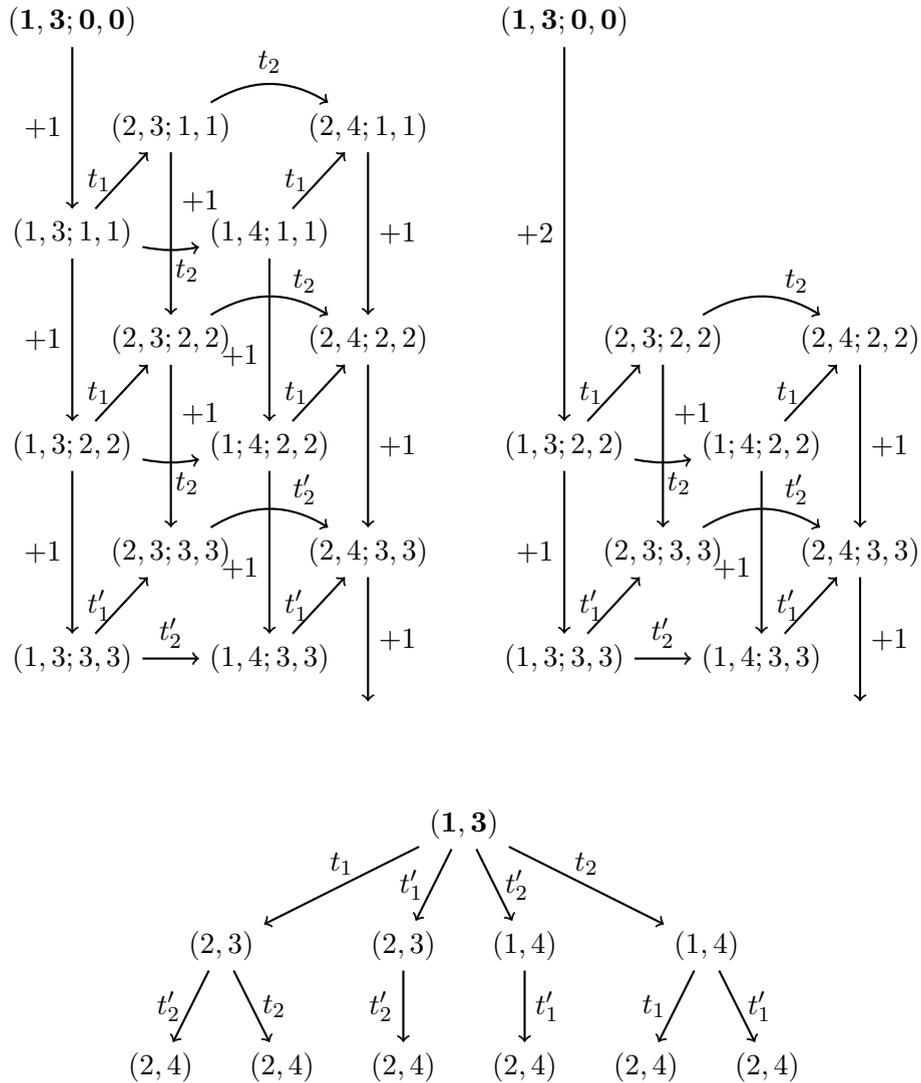
\begin{figure}[htbp]
\begin{center}

\begin{tikzpicture}[xscale=0.65, yscale=.7]
\tikzstyle{state}=[]
\tikzstyle{edge} = [->,thick]

\node [state,very thick] (q0) at (0,0) {${\bf (1,3;0,0)}$};
\node [state] (q1) at (0,-4) {$(1,3;1,1)$};
\node [state] (q2) at (0,-8) {$(1,3;2,2)$};
\node [state] (q3) at (0,-12) {$(1,3;3,3)$};
\draw [edge] (q0) to node[left, align = center]{$+1$} (q1);
\draw [edge] (q1) to node[left, align = center]{$+1$} (q2);
\draw [edge] (q2) to node[left, align = center]{$+1$} (q3);

\node [state] (q1') at (2,-2) {$(2,3;1,1)$};
\node [state] (q2') at (2,-6) {$(2,3;2,2)$};
\node [state] (q3') at (2,-10) {$(2,3;3,3)$};
\draw [edge] (q1') to node[right, pos=0.3]{$+1$} (q2');
\draw [edge] (q2') to node[right, pos=0.3]{$+1$} (q3');
\node [state] (q1'') at (4,-4) {$(1,4;1,1)$};
\node [state] (q2'') at (4,-8) {$(1;4;2,2)$};
\node [state] (q3'') at (4,-12) {$(1,4;3,3)$};
\node[](q4''')at(6,-13){};
\draw [edge] (q1'') to node[left, pos=0.6]{$+1$} (q2'');
\draw [edge] (q2'') to node[left, pos=0.6]{$+1$} (q3'');
\node [state] (q1''') at (6,-2) {$(2,4;1,1)$};
\node [state] (q2''') at (6,-6) {$(2,4;2,2)$};
\node [state] (q3''') at (6,-10) {$(2,4;3,3)$};
\draw [edge] (q1''') to node[right, align = center]{$+1$} (q2''');
\draw [edge] (q2''') to node[right, align = center]{$+1$} (q3''');
\draw [edge] (q3''') to node[right, align = center]{$+1$} (q4''');

\draw [edge] (q1) to node[left, align = center]{$t_1$} (q1');
\draw [edge] (q2) to node[left, align = center]{$t_1$} (q2');
\draw [edge] (q3) to node[left, align = center]{$t'_1$} (q3');

\draw [edge,bend right=10] (q1) to node[below, pos=0.8]{$t_2$} (q1'');
\draw [edge,bend right=10] (q2) to node[below, pos=0.8]{$t_2$} (q2'');
\draw [edge] (q3) to node[above, pos=0.5]{$t'_2$} (q3'');

\draw [edge] (q1'') to node[left, align = center]{$t_1$} (q1''');
\draw [edge] (q2'') to node[left, align = center]{$t_1$} (q2''');
\draw [edge] (q3'') to node[left, align = center]{$t'_1$} (q3''');

\draw [edge,bend left] (q1') to node[above, align = center]{$t_2$} (q1''');
\draw [edge,bend left] (q2') to node[above, pos=0.8]{$t_2$} (q2''');
\draw [edge,bend left] (q3') to node[above, pos=0.8]{$t'_2$} (q3''');

\end{tikzpicture}
\hspace{.5cm}
\begin{tikzpicture}[xscale=0.65, yscale=.7]
\tikzstyle{state}=[]
\tikzstyle{edge} = [->,thick]

\node [state,very thick] (q0) at (0,0) {${\bf (1,3;0,0)}$};
\node [state] (q2) at (0,-8) {$(1,3;2,2)$};
\node [state] (q3) at (0,-12) {$(1,3;3,3)$};
\draw [edge] (q0) to node[left, align = center]{$+2$} (q2);
\draw [edge] (q2) to node[left, align = center]{$+1$} (q3);

\node [state] (q2') at (2,-6) {$(2,3;2,2)$};
\node [state] (q3') at (2,-10) {$(2,3;3,3)$};
\draw [edge] (q2') to node[right, pos=0.3]{$+1$} (q3');
\node [state] (q2'') at (4,-8) {$(1;4;2,2)$};
\node [state] (q3'') at (4,-12) {$(1,4;3,3)$};
\node[](q4''')at(6,-13){};
\draw [edge] (q2'') to node[left, pos=0.6]{$+1$} (q3'');
\node [state] (q2''') at (6,-6) {$(2,4;2,2)$};
\node [state] (q3''') at (6,-10) {$(2,4;3,3)$};
\draw [edge] (q2''') to node[right, align = center]{$+1$} (q3''');
\draw [edge] (q3''') to node[right, align = center]{$+1$} (q4''');

\draw [edge] (q2) to node[left, align = center]{$t_1$} (q2');
\draw [edge] (q3) to node[left, align = center]{$t'_1$} (q3');

\draw [edge,bend right=10] (q2) to node[below, pos=0.8]{$t_2$} (q2'');
\draw [edge] (q3) to node[above, pos=0.5]{$t'_2$} (q3'');

\draw [edge] (q2'') to node[left, align = center]{$t_1$} (q2''');
\draw [edge] (q3'') to node[left, align = center]{$t'_1$} (q3''');

\draw [edge,bend left] (q2') to node[above, pos=0.8]{$t_2$} (q2''');
\draw [edge,bend left] (q3') to node[above, pos=0.8]{$t'_2$} (q3''');

\end{tikzpicture}\\[1cm]

\begin{tikzpicture}[scale=0.8]
\tikzstyle{edge} = [->,thick]

\node  (q0) at (0,0) {${\bf (1,3)}$};
\node  (q1) at (-1,-2) {$(2,3)$};
\node  (q2) at (-4,-2) {$(2,3)$};
\node  (q3) at (4,-2) {$(1,4)$};
\node  (q4) at (1,-2) {$(1,4)$};
\node  (q5a) at (-1,-4){$(2,4)$};
\node  (q5b) at (-3,-4){$(2,4)$};
\node  (q5b') at (-5,-4){$(2,4)$};
\node  (q5c) at (3,-4){$(2,4)$};
\node  (q5c') at (5,-4){$(2,4)$};
\node  (q5d) at (1,-4){$(2,4)$};
\draw [edge] (q0) -- node[left,midway] {$t'_1$} (q1);
\draw [edge] (q0) -- node[above,midway] {$t_1$} (q2);
\draw [edge] (q0) -- node[above,midway] {$t_2$} (q3);
\draw [edge] (q0) -- node[right,midway] {$t'_2$} (q4);
\draw [edge] (q1) -- node[left,midway] {$t'_2$}  (q5a);
\draw [edge] (q2) -- node[right,midway] {$t_2$} (q5b);
\draw [edge] (q2) -- node[left,midway] {$t'_2$} (q5b');
\draw [edge] (q3) -- node[left,midway] {$t_1$} (q5c);
\draw [edge] (q3) -- node[right,midway] {$t'_1$} (q5c');
\draw [edge] (q4) -- node[right,midway] {$t'_1$} (q5d);

\end{tikzpicture}

\end{center}
\caption{\label{fig:accelerated-space}  
The initial fragments (without variable values) 
of the various dynamics for the \model{} from Ex~\ref{ex:mapt}.
Top left: the original dynamics.
Top right: the accelerated one.
Bottom: the abstracted dynamics.}
\end{figure}

\begin{proposition} \label{abstr.prop}
The original and accelerated semantics of a \model{} lead to the same abstracted dynamics. 
\end{proposition}

\begin{proof}
We only have to show that the set of (untimed) projections of evolutions in the original semantics is the same as the ones in the accelerated one.

First, we may observe that each evolution in the accelerated semantics is also an evolution in the original one: a time passing of $\delta$ time units to reach a maximal action zone 
is the same as $\delta$ time passings of $1$ time unit;
indeed, by definition, $\delta\leq B$ and at the end in both cases we have $B-\delta=B-\delta\cdot 1 \geq 0$. 

It thus remains to show that, if $\word(\omega)$ is the projection of some evolution $\omega$ of the original semantics, it is also the projection of some evolution $\omega'$ of the accelerated one.
We shall proceed by induction on the length of $\omega$ and show more exactly that for each $\omega$ there is an accelerated evolution $\omega'$ such that $\word(\omega')=\word(\omega)$ and the set of enabled transitions/resets after $\omega$ is included in the one after $\omega'$.

The property is trivially satisfied initially, when $\omega=\omega'=\leer$,
but also if the initial enabled set is not maximal and we choose 
the accelerated strategy going to (any point realising) the first maximal enabled set through some shift $\delta$. 
Indeed, we know by definition that some shift $\delta$ always lead to the first maximal enabled set and nowhere else. 
Then, in this last case, by definition $\delta\leq B$ and the set of enabled transitions/resets increases.

We already observed that, if $\word(\omega)=\word(\omega')$, the locality and the variable are the same after $\omega$ and $\omega'$.
Let $\Delta(\omega)$ be the time elapsed during the evolution described by $\omega$.
We may observe that the clocks are determined by $\word(\omega)$ and $\Delta(\omega)$, independently on when the time passings exactly occurred: 
for any agent $A_i$, $c_i=\init_i+\Delta(\omega)
-E_i\cdot\#_{r_i}(\omega)$, where $\#_{r_i}(\omega)$ is the number of resets of $A_i$ in $\word(\omega)$.
We also have that we may not let more than $\min_i\{E_i\}$ time passings to occur in a row, since then we should have a reset occurring before.

We shall now assume that $\widetilde{\omega}$ extends $\omega$ by one event, that $\omega$ and $\omega'$ form an adequate pair, and that it is then possible to build an adequate accelerated evolution $\widetilde{\omega}'$.

If $\widetilde{\omega}=\omega(+1)$,
i.e., if $\widetilde{\omega}$ is obtained from $\omega$ by adding a time passing (of $1$ time unit), 
the projection of $\widetilde{\omega}$ is the same as the one of $\omega$, hence of $\omega'$ by the induction hypothesis. 
If the enabled set after $\widetilde{\omega}$ is still included in the one after $\omega'$, the latter still satisfies the induction hypothesis.
If the enabled set after $\widetilde{\omega}$ is no longer included in the one after $\omega'$, 
that means we reached one or more $a$'s which were not reached yet by $\omega'$, 
so that we may deduce that $\Delta(\omega')<\Delta(\widetilde{\omega})$. 
But then, going to (any point in) the next maximal action zone (with the aid of some aggregated time passing $\delta$) 
in the accelerated semantics, we shall reach those $a$'s (without trespassing $B$ since otherwise this would also occur for $\widetilde{\omega}$, forcing to first perform a transition or reset after $\omega$) 
and recover the induction hypothesis.

If $\widetilde{\omega}=\omega r_i$, for some agent $A_i$, 
we must have that $\init_i+\Delta(\widetilde{\omega})=\init_i+\Delta(\omega)=k\cdot E_i$ for some factor $k$, with $r_i$ belonging to the set of transitions/resets enabled after $\omega$. 
But an action zone enabling a reset 
is an interval including exactly one time unit, 
and is maximal. 
Hence after $\omega'$ we have the same action zone and 
$\init_i+\Delta(\omega')=k\cdot E_i$
(the same factor for $\omega'$ as for $\widetilde{\omega}$ since the time passings between resets are limited). 
We may thus also perform $r_i$ after $\omega'$, 
the state after $\omega' r_i$ is the same as after $\widetilde{\omega}$,
and the situation is the same as initially.

If $\widetilde{\omega}=\omega t$, 
for some transition $t$ of some agent $A_i$,
by the induction hypothesis $t$ may also occur after $\omega'$
and $\word(\widetilde{\omega})=\word(\omega' t)$. 
Any $t'$ enabled after $\omega$ in any $A_j$ for $j\neq i$ remains enabled after $\widetilde{\omega}$ as well as after $\omega' t$, by the induction hypothesis.
For agent $A_i$, from the third item of Constraint~\ref{constr:liveness}, 
no transition at the new location has already reached its enabling end point
in the original (after $\widetilde{\omega}$) and in the accelerated (after $\omega 't$) semantics.
If $\Delta(\widetilde{\omega})\leq\Delta(\omega 't)=\Delta(\omega')$, the clock $C_i$ of $A_i$ is not greater after $\widetilde{\omega}$ than after $\omega't$) (see the formula above yielding $c_i$) so that all the enabled transitions of $A_i$ after $\widetilde{\omega}$ are also enabled after $\widetilde{\omega}'t$, and the induction hypothesis remains valid.
On the contrary, if $\Delta(\widetilde{\omega})>\Delta(\omega't)$, it may happen that some $\widetilde{t}$ in $A_i$ is enabled after $\widetilde{\omega}$ but not after $\omega't$;
however, from $\omega't$ it is then possible to let time pass during $\Delta(\widetilde{\omega})-\Delta(\omega't)$, which leads to the same state as after $\widetilde{\omega}$;
it is then also possible to consider a maximal action zone after $\omega't$ which encompasses all the transitions enabled after $\widetilde{\omega}$, to reach it in the accelerated semantics, 
and the induction hypothesis remains valid.

\end{proof}

\section{Layers and strong and weak variables}
\label{layers}

When model checking a system, one usually has the choice between a depth-first and a width-first exploration of the state space.
For reachability properties (where one searches if some state satisfying a specific property may be reached), 
depth-first (directed and limited by the query) is usually considered more effective.
However, 
the majority of the non-determinism in systems featuring a high level of concurrency (such as \model{}s, and in particular \cav{} systems) leads to diamonds.
Indeed, if transitions on different agents are available at a state then they may occur in several possible orders, all of them converging most of the time to the same state (see the paragraph on persistence above).
In order to avoid exploring again and again the same states, a depth-first exploration needs to store 
all the states already visited up to now, 
which is usually impossible to do in case of large systems. 
For example, if states $s1$ and $s2$ share a common successor $s3$, the algorithm will compute successors of $s1$, then remove $s1$ from memory and continue with its successors, until reaching $s3$ and exploring all paths from $s3$, forgetting each time the nodes already visited. 
That way, when the algorithm has explored all paths from $s1$ and start exploring from $s2$, there is no memory of $s3$ having been explored already, and thus all paths starting from it will be explored again.

On the contrary, using width-first algorithms would guarantee avoiding that issue, 
because duplicate states obtained at a given depth can be removed. 
However, this would also imply exploring all reachable states at a given depth and forbid using heuristics to direct and limit the exploration.

An idea is then to try to combine both approaches.

\subsection{Layered state space} 
The state space of a \model{} shows an interesting characteristics: apart from having no cycles (see Prop.~\ref{DAG.prop}: the state space in our case is always a \acygra), 
its structure can often be divided in 
layers such that all states on the border of a layer share the same vectors of localities and clocks (and thus, the same set of enabled transitions) and are situated at the same time distance from the initial state.
The only difference concerns the value of the variable, due to the non-determinism and the concurrency inherent to this kind of models. 
Non-determinism means that an agent has the choice between several transitions at some location; concurrency means that at least two agents may perform transitions at some point.
In the first case, several paths may be followed by the agent to reach some point, leading to different values of the variable; 
in the second case, transitions of the two agents may be commuted, leading again to different values of the variable.

This is schematised by 
Fig.~\ref{fig:layered_space}, where one can see how the state space is divided in sub-spaces 
(each of them being a \acygra{} with a unique initial state) such that each final state of a sub-space 
is the initial state of another one.
The sub-spaces may intersect.

More formally, 
in a \acygra, we have a natural partial order:
$s_1 < s_2$ if there is a non-empty path from $s_1$ to $s_2$;
$s_1$ and $s_2$ are incomparable if there is no non-empty path between them.

A cut is a maximal subset of incomparable states.
A cut partitions the partially ordered space into three subsets: the states before 
the cut, the cut itself, and the states after the cut.

In the following, we shall denote by $\omega$ a (possibly empty) evolution leading from some state $s$ to some state $s'$, \ie, the sequence of transitions, resets and time passings labelling some path going from $s$ to $s'$ in the state space of the considered model. As usual, we shall also denote by $\Delta(\omega)$ the sum of the time passings along $\omega$, also called the time distance from $s$ to $s'$ (along $\omega$). 

\begin{definition}
In a \model{} (whose state space is a \acygra), a cut is said coherent if all its states have the same locality and clock vectors, and any two evolutions $\omega_1$, $\omega_2$ linking the initial state to states of the cut have the same time length: $\Delta(\omega_1)=\Delta(\omega_2)$. Coherent cuts may be used to define borders between layers.

Let $s$ be any state in a \model;
the states reachable from $s$ form a \acygra{} subspace, in which we may also define coherent cuts:
a coherent cone with apex $s$ is the set of states up to a coherent cut in this subspace (including $s$ and the cut). 
\end{definition}

\begin{proposition}
\hspace*{1cm}
\begin{itemize}
\item Coherent cuts do not cross, in the following sense.
Let $\mathcal C_1$ and $\mathcal C_2$ be two coherent cuts in a \model,
$s_1,s'_1\in\mathcal C_1$,
$s_2,s'_2\in\mathcal C_2$.
If there is an evolution $\omega$ from $s_1$ to $s_2$ and  $\omega'$ from  $s'_2$ to $s'_1$, then $\mathcal C_1=\mathcal C_2$ and $\omega=\leer=\omega'$. In particular, no two distinct coherent cuts may have a common state. 
\item The time distance between coherent cuts is constant, in the following sense.
Let ${\mathcal C}_1$ and ${\mathcal C}_2$ two different coherent cuts in a \model,
    $s_1$, $s'_1\in{\mathcal C}_1$, $s_2$, $s'_2\in{\mathcal C}_2$,
    with an evolution $\omega$ from $s_1$ to $s_2$ and $\omega'$ from  $s'_1$ to $s'_2$, then $\Delta(\omega)=\Delta(\omega')$.
\item If $s_1\in{\mathcal C}_1$ and there is a coherent cone with time height $\Delta$ (for any evolution $\omega$ from $s_1$ to the base of the cone, $\Delta(\omega)=\Delta$), then there is a coherent cut ${\mathcal C}_3$ separated from ${\mathcal C}_1$ by a time distance $\Delta$.
\end{itemize}
\end{proposition} 

\begin{proof}\hspace*{1cm} 
\begin{itemize}
    \item In the first case,
 if there is a path $\widetilde{\omega}$ from $s_0$ to $s_1$ and a path $\widetilde{\omega}'$ from $s_0$ to $s'_2$, we must have $\Delta(\widetilde{\omega})+\Delta(\omega)=\Delta(\widetilde{\omega}\omega)=\Delta(\widetilde{\omega}')$ and $\Delta(\widetilde{\omega}')+\Delta(\omega')=\Delta(\widetilde{\omega}'\omega')=\Delta(\widetilde{\omega})$, hence $\Delta(\omega)=-\Delta(\omega')$, which is only possible if $\Delta(\omega)=0=\Delta(\omega')$.\\
 Also, since $s_2$ and $s'_2$ have the same localities and clocks, there is an evolution $\omega'$ from $s'_2$ to some state $s''_1$ with the same localities and clocks as $s_1$,
 hence an evolution $\omega\omega'$ of time length $0$ from $s_1$ to $s''_1$, which reproduces the same localities and clocks.
 Since the localities of each agent $A_i$ form a \acygra{} and $E_i>0$, this is only possible if $\omega=\leer=\omega'$. \\
 In particular, if $\omega=\leer$, i.e., $s_1=s_2$, we also have that $s'_1=s'_2$ and $\mathcal C_1=\mathcal C_2$, and similarly if $\omega'=\leer$.
\item
In the next case, if $s_0$ is the initial state and there is a path $\widetilde{\omega}$ from $s_0$ to $s_1$ and a path $\widetilde{\omega}'$ from $s_0$ to $s'_1$,
we must have $\Delta(\widetilde{\omega})+\Delta(\omega)=\Delta(\widetilde{\omega}\omega)=\Delta(\widetilde{\omega}'\omega')=\Delta(\widetilde{\omega}')+\Delta(\omega')$, hence the property.
\item The last property results from the observation that, if $s'_1\in{\mathcal C}_1$, since $s_1$ and $s_2$ have the same localities and clocks, any evolution from $s_1$ to the base of the cone is also present from $s'_1$ and leads to a state with the same localities and clocks as the states on the base of the cone (but the variables may differ). 
And conversely, if a path leads from $s'_1$ to a state of ${\mathcal C}_3$, it has time length $\Delta$ and there is the same path from $s_1$ to some state on the base of the cone.
\end{itemize}
\end{proof}

For instance, if agent $A_i$ 
in a \model{} starts at $l_i^1$ with a null clock, after $E_i$ time units and before ($E_i+1$) time units, it shall necessarily pass through its reset
(it is possible that it performs other transitions before and/or after this reset without modifying its clock, but it is sure the agent will go through this reset before performing a new time passing). 
Hence, if each agent starts from its initial locality with a null clock, after $\lcm\{E_1,\ldots,E_n\}$ (\ie, the least common multiple of the various reset periods; 
in the following, we shall denote this value by $\lcmE$) 
time units, it is sure we shall be able to revisit the initial state, but possibly with various values of the variable, yielding the border of a layer. From this border the same sequences of transitions/resets/time-passings as initially will occur periodically (with a period of $\lcmE$), 
leading to new borders, with the initial localities and the null clocks.

The situation will be similar if $\forall A_i:\init_i=\init \mod E_i$ for some value $\init<\lcmE$. 
Indeed, for each agent $A_i$ and each $k$, after $E_i-\init_i+k\cdot E_i$ time passing we shall visit $l_i^1$ with a null clock, hence visit new borders after $(k+1)\cdot\lcmE-\init$ time passings. 

For other initial values of the clocks, it is not sure we shall be able to structure the state space in layers, but in either case, it may also happen there are other kinds of layers and borders.

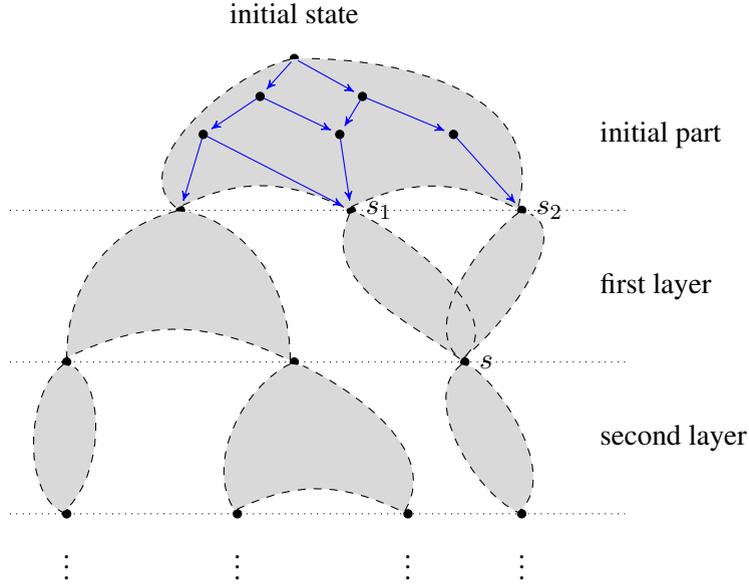
\begin{figure}[htb]
\begin{center}

\begin{tikzpicture}[>=stealth',shorten >=2pt,xscale=1.5,yscale=2]
\tikzstyle{line} = [-]
\tikzstyle{edge} = [->]
\tikzstyle{state}=[draw,circle,fill=black,scale=.3pt]

\node[state]  (s1) at (0,0) {};
\node (s0) at (0,.3) {initial state};

\node[state]  (s11) at (-1,-1) {};
\node[state, label=right:{$s_1$}]  (s12) at (0.5,-1) {};
\node[state, label=right:{$s_2$}]  (s13) at (2,-1) {};
\draw[dashed, fill=gray!30] (s1) [out=190, in=150] to  (s11) [out=20, in=160] to  (s12) [out=20, in=160] to  (s13) [out=70, in=0] to  (s1);
\draw[dashed] [line, dotted] (-2.5,-1) -- (3,-1);
\node (i0) at (2.6,-.5) [right] {initial part};

\node[state]  (a) at (-.3,-.25) {};
\node[state]  (b) at (0.6,-.25) {};
\node[state]  (c) at (-.8,-.5) {};
\node[state]  (d) at (0.4,-.5) {};
\node[state]  (e) at (1.4,-.5) {};
\draw[edge,blue] (s1) to (a);
\draw[edge,blue] (s1) to (b);
\draw[edge,blue] (a) to (c);
\draw[edge,blue] (a) to (d);
\draw[edge,blue] (b) to (d);
\draw[edge,blue] (b) to (e);
\draw[edge,blue] (c) to (s11);
\draw[edge,blue] (c) to (s12);
\draw[edge,blue] (d) to (s12);
\draw[edge,blue] (e) to (s13);
\node[state]  (s21) at (-2,-2) {};
\node[state]  (s22) at (0,-2) {};
\node[state, label=right:{$s$}]  (s23) at (1.5,-2) {};
\draw[dashed, fill=gray!30] (s11) [out=190, in=90] to (s21) [out=20, in=160] to  (s22) [out=90, in=-10] to  (s11);
\draw[draw=none, fill=gray!30] (s12) [out=240, in=150] to  (s23) [out=30, in=-20] to  (s12);
\draw[dashed, fill=gray!30] (s13) [out=210, in=150] to  (s23) [out=30, in=-20] to  (s13);
\draw[dashed] (s12) [out=240, in=150] to  (s23) [out=30, in=-20] to  (s12); 
\draw[dashed] [line, dotted] (-2.5,-2) -- (3,-2);
\node[dashed] (i1) at (2.6,-1.5) [right] {first layer};

\node[state]  (s31) at (-2,-3) {};
\node[state]  (s32) at (-0.5,-3) {};
\node[state]  (s33) at (1,-3) {};
\node[state]  (s34) at (2,-3) {};
\draw[dashed, fill=gray!30] (s21) [out=210, in=160] to  (s31) [out=30, in=-30] to  (s21);
\draw[dashed, fill=gray!30] (s22) [out=200, in=130] to  (s32) [out=20, in=160] to  (s33) [out=50, in=-20] to  (s22);
\draw[dashed, fill=gray!30] (s23) [out=210, in=160] to  (s34) [out=30, in=-30] to  (s23);
\draw[dashed] [line, dotted] (-2.5,-3) -- (3,-3);
\node[dashed] (i1) at (2.6,-2.5) [right] {second layer};

\drop{
\node[state]  (s41) at (-1,-4) {};
\node[state]  (s42) at (1,-4) {};
\node[state]  (s43) at (2,-4) {};
\draw[draw=none, fill=gray!30] (s31) [out=210, in=160] to  (s41) [out=30, in=-30] to  (s31);
\draw[dashed, fill=gray!30] (s32) [out=210, in=160] to  (s41) [out=30, in=-30] to  (s32);
\draw[dashed] (s31) [out=210, in=160] to  (s41) [out=30, in=-30] to  (s31); 
\draw[dashed, fill=gray!30] (s33) [out=210, in=160] to  (s42) [out=30, in=-30] to  (s33);
\draw[dashed, fill=gray!30] (s34) [out=210, in=160] to  (s43) [out=30, in=-30] to  (s34);
\draw[dashed] [line, dotted] (-2.5,-4) -- (3,-4);
\node[dashed] (i1) at (2.6,-3.5) [right] {third layer};
\node  (41) at (-1,-3.3) {$\vdots$};
\node  (42) at (1,-3.3) {$\vdots$};
\node  (43) at (2,-3.3) {$\vdots$};
} 

\node  (41) at (-2,-3.3) {$\vdots$};
\node  (42) at (-.5,-3.3) {$\vdots$};
\node  (43) at (1,-3.3) {$\vdots$};
\node  (43) at (2,-3.3) {$\vdots$};

\end{tikzpicture}
\end{center}
\caption{\label{fig:layered_space} 
General shape of a layered state space with a zoom on the initial part.
Identical states are merged together.
All states on the border of a layer share the same vectors of localities and clocks but have different values of $v$. 
Each sub-space surrounded by dashed lines 
is a \acygra{} having a unique initial state and one or more final states as shown in blue for the sub-space corresponding to the initial part.} 
\end{figure}

\begin{figure*}[tb]
\centering
\tikzstyle{location}=[draw,circle,inner sep=2pt]
\begin{tikzpicture}[>=latex',xscale=.8, yscale=1,every node/.style={scale=1}]
\node[location] at (0,0) (a) {$l_1^1$};
\node[location] at (3,0) (b) {$l_1^2$};
\node[location] at (6,0) (c){$l_1^3$};
\node at (8.5,0) {$E_1 = 10$};
\node at (-2,0) {\begin{tabular}{l} $A_1$ \end{tabular}};
\draw[->,rounded corners] (a) -- node[above,midway] {$[1,5]$} (b);
\draw[->,rounded corners] (b) -- node[above,midway] {$[6,8]$} (c);
\end{tikzpicture} \\[.5cm]
\begin{tikzpicture}[>=latex',xscale=.9, yscale=1,every node/.style={scale=1}]
\node[location] at (0,0) (e) {$l_2^1$};
\node[location] at (2,0) (f) {$l_2^2$};
\node[location] at (4,-1) (g) {$l_2^3$};
\node[location] at (6,0) (h) {$l_2^4$};
\node at (8,0) {$E_2 = 15$};
\node at (-1.5,0) {\begin{tabular}{l} $A_2$ \end{tabular}};
\draw[->,rounded corners] (e) -- node[above,midway] {$[0,4]$} (f);
\draw[->,rounded corners] (f) -- node[below,sloped,midway] {$[6,7]$} (g);
\draw[->,rounded corners] (g) -- node[below,sloped,midway] {$[9,9]$} (h);
\draw[->,rounded corners] (f) --  node[above,midway] {$[7,11]$}  (h);
\end{tikzpicture} \\[.5cm]
\includegraphics[scale=0.35]{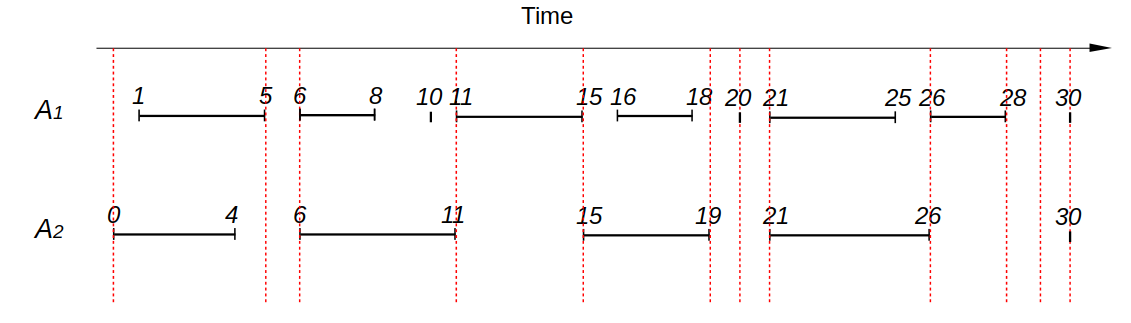}

\caption{\label{fig:intervals} Top: Example of a \model{} composed of agents $A_1$ and $A_2$ with clocks $C_1$ and $C_2$ initialized to $0$. 
Bottom: Time intervals where a transition or set of transitions may be performed. 
Red dotted lines indicate time units where a coherent cut may exist.}
\end{figure*}

Let us consider for instance the system illustrated on top of Fig.~\ref{fig:intervals}, where each agent starts with a null clock in its initial locality.
Agent $A_1$ is deterministic since there is a single transition originated from $l_1^1$ as well as from $l_1^2$, and agent $A_2$ is not since two transitions may occur while being in $l_2^2$. 
If we forget the value of the variable, 
the graph for $A_1$ is periodic with a period of $E_1=10$ and the graph for $A_2$ is periodic with a period of $E_2=15$. 
The whole system is therefore periodic with a period of $\lcmE=\lcm\{E_1,E_2\} = 30$. 
The sequence of intervals depicted in the bottom of the figure for $A_1$, represents the intervals where transitions
have to take place, with their time distance from the initial state (here these intervals are disjoint, but they could overlap as well). 
Note that if an initial clock $\init_i$ were to be strictly positive, the sequence of intervals for agent $i$ would be shifted to the left by $\init_i$ time units. 
For $A_2$, the intervals exhibited on the figure have a different interpretation, that will be explained below.

A border may not occur at a time $t$ measured from the beginning of the system if,
for some deterministic agent $A_i$ (the case for non-deterministic agents will be handled below), 
there is an interval $[a,b]$, shifted by some multiple of $E_i$, such that 
\begin{equation}\label{eq1}
a+k\cdot E_i-\init_i < t < b+k\cdot E_i-\init_i\end{equation} 
Indeed, in that case, at time $t$, $A_i$ may either be at the source or at the destination of the corresponding transition, without being able to impeach that, hence without being certain of the locality.

Hence, a deterministic agent may allow a border to occur at a time $t$ if one of the following cases occurs:
\begin{itemize}
    \item If $t$ is strictly between the various shifted intervals of $A_i$, we know immediately that when we reach this time we are at some specific location in $A_i$. 
    For instance at time $t=19$ in Fig.~\ref{fig:intervals}, we are sure $A_1$ is in location $l_1^3$.
    \item If $t$ is situated at the right of some (shifted) interval, the agent can be either in the source or in the destination localities.
    The first situation does not exist in every paths, as it is possible to leave the source before $t$, while the second situation exists in all paths.
    Therefore, the second situation is suitable for a coherent cut.
    This case happens in Fig.~\ref{fig:intervals}, for instance when the system reaches time $t=5$, $A_1$ may be either in $l_1^1$ or in $l_1^2$.
    \item If $t$ is situated at the left of some (shifted) interval, the agent can be either in the source or in the destination localities.
    This is symmetric to the previous case, and here it is the first situation that exists in all paths and is suitable for a coherent cut.
    This case happens in Fig.~\ref{fig:intervals}, for instance when the system reaches time $t=6$, $A_1$ may be either in $l_1^2$ or in $l_1^3$.
    \item If $t$ is situated on an interval of length $0$ (such as a reset, or transition with an interval where $a=b$).
    This corresponds to a union of the two previous cases, where two localities are possible.
    Here, both are suitable for a coherent cut.
    This case happens in Fig.~\ref{fig:intervals}, for instance when the system reaches time $t=20$, $A_1$ may be either in $l_1^3$ or in $l_1^1$.
    \item If $t$ is both at the right of some shifted interval and at the left of another one (meaning that they intersect on $t$), this comes back to a combination of the previous cases. 
    As such, a suitable situation for a coherent cut is to consider the system after performing the transition corresponding to the left interval and before performing the transition corresponding to the right interval.
    A particular occurrence of this case is shown in Fig.~\ref{fig:intervals} at time $15$, where $A_2$ 
    is at the right of an interval of length $0$ corresponding to its reset and at the left of the interval of the transition from $l_2^1$.
    In this situation, $A_2$ may either be in $l_2^4$, $l_2^1$ or $l_2^2$.
    The fact that one of the interval is of length $0$, is included in the general case.
\end{itemize}

For a non-deterministic agent, like $A_2$ in Fig.~\ref{fig:intervals}, the analysis is similar but 
slightly more complex; indeed, even between intervals it may be in several possible localities\footnote{of course not at the same time: for different histories.
}. 
For instance, at time $t=7$, $A_2$ may either be in $l_2^3$ after having  performed a transition at time $6$, or in $l_2^2$, and we may not force the system to wait for $A_2$ going in $l_2^3$ since it has the possibility to choose the other transition.
The idea is then to consider the localities of the considered non-deterministic agent $A_i$ which by themselves are singleton cuts in the \acygra{} of its localities, 
\ie, the localities which are visited in every complete iteration (from $l_i^1$ to $l_i^m$). 
For $A_2$ in Fig.~\ref{fig:intervals}, those localities are $l_2^1$, $l_2^2$ and $l_2^4$.
They form a sequence in $L_i$: let $P_i$ be this list 
and denote by $\nextt(l)$ the successor of $l$ in $P_i$.
Thus, between two localities $l$ and $\nextt(l)$ in $P_i$, either there is a unique transition 
enabled in some interval $[a,b]$ (as in the deterministic case above) or there are at least
two different paths with possibly several transitions enabled at some moment in the interval $[\tilde{a},\tilde{b}]$, where $\tilde{a}$ is
the smallest lower bound of all the outgoing transitions from $l$ and 
$\tilde{b}$ is the greatest upper bound of all the incoming transitions to $\nextt(l)$.
One may think about $[\tilde{a},\tilde{b}]$ as the enabling interval of some virtual transition from $l$ to $\nextt(l)$.
Then, exactly the same argument as above may be used to check if a state space of a \model{} admits layers,
and this amounts to a proof of: 

\begin{proposition}  \label{cutint.prop}
A \model{} admits a layered state space with borders at $t+\ell\cdot\lcmE$  
with $\ell \in \mathbb{N}$ if, 
for each agent $A_i$, 
for each $l \in P_i\setminus\{l_i^m\}$, 
and for $\tilde{a} \defeq \min\{a \mid (l,f,[a,b],l')\in \post{l}\}$, 
$\tilde{b} \defeq \max\{b \mid (l',f,[a,b],\nextt(l))\in \pre{\nextt(l)}\}$, 
no $k \in \mathbb{N}$ satisfies: 
$\tilde{a}+k\cdot E_i-\init_i < t < \tilde{b}+k\cdot E_i-\init_i$,
where
$\init_i$ is the initial value of clock $C_i$. \hfill $\Box$ 
\end{proposition} 

A border detected that way will then be defined by the couple $((l_1,\cdots,l_n)$, $(c_1,\cdots,c_n))$ where, for agent $A_i$, $c_i = (t+\init_i) \mod E_i$ 
(or $E_i$ instead of $0$ if we reach the position of a reset but decide not to perform the latter) 
and $l_i$ is the locality periodically reached in all possibles paths at clock value $c_i$.
Locality $l_i$ is determined by the clock if we are not at the border of an interval, otherwise we have to know if the corresponding transition or reset has to be performed. 
In particular, when reaching intervals of length $0$, we have a choice between several localities (before or after performing the corresponding transition or reset).

\paragraph{Searching for $t$ more efficiently}
Proposition~\ref{cutint.prop} allows to search for the coherent cuts whose sets of localities and clocks reproduce every $\lcmE$ time units\footnote{There may also be non-periodic coherent cuts at the beginning of the state space, if some agents do not start at their initial locality with a null clock. Indeed, for those agents, it may happen that other localities are certainly visited before the first reset, which introduce other intervals $[\tilde{a},\tilde{b}]$ before that time. However, we shall not use those extra coherent cuts in our exploration and model checking tool.}. 
Hence, it is not necessary to consider times $t$ beyond $\lcmE$;
note however that it may happen that a coherent cut occurs at time $\lcmE$, but not at time $0$, if the corresponding localities occur "before" the initial ones at time $\lcmE$. 
Also, this proposition seems to imply we should consider all the shifted version of each interval $[\tilde{a},\tilde{b}]$, \ie, all integer values for $k$.
This is not true: for each $[\tilde{a},\tilde{b}]$ we only have to consider the greatest $k$ respecting the left constraint 
$\tilde{a}+k\cdot E_i-\init_i < t$, 
\ie, the greatest $k_a$ such that $k_a< \frac{t+\init_i-\tilde{a}}{E_i}$, which is given by the formula $k_a=\lceil\frac{t+\init_i-\tilde{a}}{E_i}-1\rceil$.
We then have to check if $t<\tilde{b}+k_a\cdot E_i-\init_i$ (in which case the considered $t$ does not define a coherent cut).

If we also want to avoid the extremities of the intervals $[\tilde{a},\tilde{b}]$, we get that no $k$ should lead to the constraint 
$\tilde{a}+k\cdot E_i-\init_i \leq t \leq \tilde{b}+k\cdot E_i-\init_i$.
This leads to the simpler formula $k_a=\lfloor\frac{t+\init_i-\tilde{a}}{E_i}\rfloor$, and to the check $t\leq\tilde{b}+k_a\cdot E_i-\init_i$.
Also, in this case the clock vector is enough to describe the coherent cut without any ambiguity.

\paragraph{Combination with the accelerated semantics} 

If we consider only the locality and clock vectors and we neglect the value of variable $V$ in the states, a coherent cut becomes a mandatory crossing point in the original dynamics of the system.
This will also be true in the accelerated semantics, but in order to preserve the periodic occurrences of these points we need to avoid letting time jumps go anywhere in the next maximal action zone: we need a deterministic rule, like the one we mentioned before, prescribing to go to the end of the zone. We shall adopt this rule in the following. 

Since in the accelerated semantics, time passings jump to (the end of) the next maximal action zone, intervals do not play the same role as in the original semantics and we may not rely on Property~\ref{cutint.prop} to find the coherent cuts. 
In particular, coherent cuts in the accelerated semantics are usually not ones in the original one.
This is due to the fact that, as time steps may be bigger than one unit in the accelerated dynamics, it may happen that a time passing overpasses the clock vector corresponding to some coherent cut $(\vec{l},\vec{c})$ present in the original dynamics.

However, we may relate 
coherent cuts in the accelerated semantics to the ones in the original one, which may be characterised by Property~\ref{cutint.prop}:
as we shall see in Proposition~\ref{acc-cuts.prop}, 
a state $(\vec{l},\vec{c}+\delta)$ reached after going over an original coherent cut $(\vec{l},\vec{c})$ is in fact a coherent cut of the accelerated dynamics.
This is illustrated in Figure~\ref{fig:accel-cut}.

One may observe that in both dynamics, all paths go to either $((1,2);(5,5))$ or $((2,2);(4,4))$.
In the original dynamics, the coherent cut at $((2,2);(5,5))$ is reached, and after a time passing the coherent cut at $((2,2);(6,6))$ is reached.
From $((2,2);(6,6))$, three actions are possible (two transitions and one time passing).
In the accelerated dynamics, it is still possible from $((1,2);(5,5))$ to reach $((2,2);(5,5))$, but not from $((2,2);(4,4))$ as the acceleration directly leads to $((2,2);(7,7))$.
From $((2,2);(5,5))$ in the accelerated semantics, the acceleration also leads to $((2,2);(7,7))$, since this state corresponds to the end of the first maximal action zone.
As such, in the accelerated semantics, $((2,2);(6,6))$ is not a coherent cut anymore since it is not reachable, but also $((2,2);(5,5))$ is no longer a coherent cut since there exist paths that go over it.
This illustrates that, in the accelerated dynamics, the locality vector corresponding to an original cut may be entered with different clock values, but from those states (here $((2,2);(4,4))$ and $((2,2);(5,5))$) the acceleration will always lead to the same vectors of localities and clocks (here $((2,2);(7,7))$), which is a coherent cut in the accelerated semantics.

It remains to show that this is not an accident but a general rule.

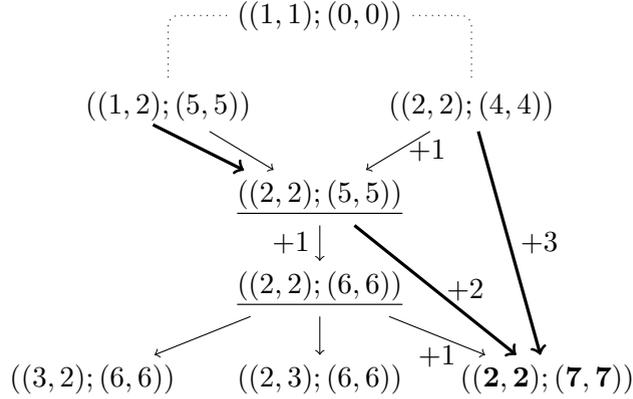
\begin{figure}
\begin{center}

\begin{tikzpicture}[yscale=1.2]
\tikzstyle{edge} = [->]
\tikzstyle{bedge} = [->,very thick] 
\tikzstyle{cut}=[draw,rectangle,rounded corners,inner sep=2pt]

\node  (q0) at (0,-1) {$((1,1);(0,0))$};
\node  (q1a) at (-2,-2) {$((1,2);(5,5))$};
\node  (q1b) at (2,-2) {$((2,2);(4,4))$};
\node  (q2) at (0,-3) {$\underline{ ((2,2);(5,5))}$};
\node  (q3) at (0,-4) {$\underline{((2,2);(6,6))}$};
\node  (q4a) at (-3,-5){$((3,2);(6,6))$};
\node  (q4b) at (0,-5){$((2,3);(6,6))$};
\node  (q4c) at (3,-5){$\bf ((2,2);(7,7))$};
\draw  [dotted,rounded corners](q0) -|  (q1a);
\draw  [dotted, rounded corners](q0) -|  (q1b);
\draw [edge] (q1a) -- (q2);
\draw [edge] (q1b) -- node[right, midway] {$+1$} (q2);
\draw [edge] (q2) -- node[left, midway] {$+1$} (q3);
\draw [edge] (q3) -- (q4a);
\draw [edge] (q3) -- (q4b);
\draw [edge] (q3) -- node[below, midway] {$+1$} (q4c);

\draw [bedge] (-2.2,-2.2) -- (-1,-2.7);
\draw [bedge] (q1b) -- node[right, midway] {$+3$} (q4c);
\draw [bedge] (q2) -- node[right, midway] {$+2$} (q4c);

\end{tikzpicture}

\end{center}    
\caption{\label{fig:accel-cut} A fragment of the original and accelerated dynamics with omitted values of $V$ for Example \ref{fig:intervals}. Vectors of localities $(l_1^i,l_2^j)$ are denoted by $(i,j)$.
The thick arcs correspond to the steps present in the accelerated dynamics while thin ones correspond to the steps present in the original one. Time passing arcs are labelled by the corresponding delay; transition arcs are unlabelled (the corresponding transition may be read in the change of localities). Coherent cuts in original dynamics are underlined and those in the accelerated one are bold. }
\end{figure}

\begin{proposition} \label{acc-cuts.prop}
For each (periodic, with the period $\lcmE$) coherent cut 
characterised by the vectors $(\vec{l},\vec{c})$ at time $t$ (measured from the beginning of the system) in the original semantics, 
there is a coherent cut in the accelerated semantics for the same vector of localities $\vec{l}$ and clock vector $\vec{c} + \delta$ at time $t+\delta$,
for some $\delta \in \mathbb{N}$. 
\label{prop:accel-cut}
\end{proposition} 

\begin{proof}

Since, in the accelerated semantics as in the original one, the localities are determined by the sequence of transitions and resets that have been performed, from Proposition~\ref{abstr.prop} we know that the visited localities are the same in both semantics.
Moreover, each reachable state in the accelerated dynamics is also reachable in the original one and for each existing path between two states in the accelerated dynamics there is also at least one path in the original one.
As a coherent cut is a mandatory crossing point (when we neglect the values of the variable) in the original dynamics of the system, the only way to avoid it in the accelerated dynamics is to have a new arc from a state before the cut leading to a state after the cut
(for instance, in Figure~\ref{fig:accel-cut}, the original cut $(2,2)(5,5)$ is reachable in the accelerated semantics, but it may also be skipped by the arc from $(2,2)(4,4)$ to $(2,2)(7,7)$, hence it is not a cut in the accelerated semantics; 
the original cut $(2,2)(6,6)$ is not even reachable in the accelerated semantics, due to the arc from $(2,2)(5,5)$ to $(2,2)(7,7)$).
All transitions and resets present in the accelerated dynamics are also present in the original one, therefore only a time passing (jumping to the end of the next maximal action zone) may provide such a possibility. 
Hence, if we may prove that whenever a time passing in the accelerated dynamics goes from a state $s$ before a cut in the original dynamics to a state $s'$ after that cut, the state $s'$ belongs to a coherent cut in the accelerated dynamics, we are done. 

If an agent $A_i$ has a single location $l_i^1$, i.e., $m_i=1$, its resets do not change the location (only its clock goes from $E_i$ to $0$), 
hence we shall neglect it in the following definition of $t^-$ and $t^+$, considering its resets are spurious.
Let $t^- = \max\{t'\mid t'= (k\cdot E_i -\init_i)\leq t, k>0, i\in \{1,\ldots, n\}, m_i>1\}$ 
be the time of the last (non-spurious) reset not after $t$, 
and $t^+ = \min\{t'\mid t'= (k\cdot E_i-\init_i) \geq t, k>0, i\in \{1,\ldots, n\},m_i>1\}$ 
be the time of the first (non-spurious) reset not before $t$.

If the original coherent cut $(\vec{l},\vec{c})$ occurs before the first non-spurious reset then,
with the usual convention $\max(\emptyset)=0$ in $\mathbb{N}$,
$t^- = 0$. 
If a reset is available or was just performed at $t$, then $t^- = t = t^+$. 

Since we assumed that the transition graph of each agent is acyclic, 
if an agent leaves a locality, the same locality cannot be reached again before the next reset of this agent.
As a consequence, in the interval $[t^-,t^+]$, a vector of localities $\vec{l}$ once exited (\ie, performing a transition from a state with $\vec{l}$) cannot be reached again.
Therefore in both semantics, it is not possible to enter $\vec{l}$ strictly after $t$ in the interval $[t^-,t^+]$, 
nor to leave $\vec{l}$ strictly before $t$ in the interval $[t^-,t^+]$, 
since $\vec{l}$ must be reached at time $t$ in each original path (by definition of a coherent cut in the original semantics; 
note that other vectors of localities may also be reached at $t$, before or after $\vec{l}$). 

Hence, in the accelerated semantics, $\vec{l}$ will always be entered at some $t' \leq t$ and may only be leaved at some $t''\geq t$. 
There may be several values for $t'$, depending on the path followed to reach this locality (for instance, in Figure~\ref{fig:accel-cut}, there are two ways to enter $(2,2)$: $(2,2)(4,4)$ and $(2,2)(5,5)$).
On the contrary, there is single value $t''$, corresponding to the end of the first maximal action zone starting at or after $t$, and there is one since otherwise that would mean there is no way to reach $t$ and get out of $\vec{l}$.
This  yields the unique way to get out of the locality vector $\vec{l}$, hence a coherent cut of the accelerated semantics, adding $\delta=t''-t$ to each clock since we did not performed a reset meanwhile.
\drop{
If $t' < t$, since $\vec{l}$ cannot be exited before $t$, only a time passing of $d$ time units is possible, such that $t' + d \geq t$.
If $t' = t$, the coherent cut is either reached with $\delta = 0$, or only a time passing of $d > 0$ time units is possible that leads to the coherent cut.}
For instance, in the example of Figure~\ref{fig:accel-cut}, if $t=5$ there are two possible paths, either $t'=4$ and $d=3$, or $t'=5$ and $d=2$, leading in both cases to the coherent cut $(2,2)(7,7)$ of the accelerated semantics. 
Notice a curious feature: 
in the accelerated semantics for the same example, we reach the coherent cut $((2,2);(5,5))$ of the original semantics, but it is no longer a coherent cut since there exists now a path that does not reach it, because of the added arc labelled $+3$.

From the choice of the jump points in the accelerated semantics, $\delta$ will be the same for each re-occurrence of the considered coherent cut, at $t+k\cdot\lcmE$.

\drop{
Therefore only a time passing is possible from $((2,2);(5,5))$, and we are in the case were the coherent cut is reached after this next time passing. 

The time passing of $d$ time units acts as $d$ time passings of $1$ time unit.
This time passing jumps over a state that belongs to a coherent cut in the original semantics, and all paths from that cut lead at $t'+d$ to a set of states which all possess the same vector of localities and clocks, as the accelerated semantics always jump to the end of the first maximal action zone.
By definition this set of states if a coherent cut as they all have the same vector of localities $\vec{l}$ and clocks $\vec{c}+\delta$, with $\delta = t' + d - t$. }
\end{proof}

\paragraph{Exploring layered state space}

The function $next\_border(\state)$, depicted in Algorithm~\ref{algo:next_border} takes a state $\state=(\vec{l},\vec{c},v) $ and computes, through a width first exploration, the set of successors up to the next border. 
It applies to both original and accelerated semantics and requires to define a non empty set of periodic cuts $\mathit{Cuts}$ (in the form $(\vec{l},\vec{c})$, i.e., without the variable, obtained from an application of Prop.~\ref{cutint.prop}) that are coherent in the original dynamics.

To do so we introduce the function $next\_state(s)$, which returns the set of all successors of state $s$ (depending on the chosen semantics), 
and the function $is\_cut(pre\_s,s)$, which is true if the state $s$, successor of state $pre\_s$ is part of a cut defined by $\mathit{Cuts}$.
Formally, $is\_cut(pre\_s,s)$ depends on the chosen semantics.
In the original semantics, $is\_cut(pre\_s,s)$ is $\true$ if $s=(\vec{l},\vec{c},v) $ and $ (\vec{l},\vec{c})\in \mathit{Cuts}$.
In the accelerated semantics, $is\_cut(pre\_s,s)$ is $\true$ if one of the following occurs:
\begin{itemize}
    \item $s=(\vec{l},\vec{c},v) $, $ (\vec{l},\vec{c})\in \mathit{Cuts}$ and at least one transition or reset allows to leave $s$, 
    which means that the coherent cut is the same in both semantics;
    \item $pre\_s=(\vec{l},\vec{c},v)$, $s=(\vec{l},\vec{c^+},v) $, $ (\vec{l},\vec{c})\in \mathit{Cuts}$ and $s$ is the only successor of $pre\_s$ with $\vec{c} < \vec{c^+}$, which means that the original cut has also been reached in accelerated semantics but is no longer a coherent cut;
    \item $pre\_s = (\vec{l},\vec{c^-},v)$, $s = (\vec{l},\vec{c^+},v)$ and $(\vec{l},\vec{c})\in \mathit{Cuts}$ with $\vec{c^-} < \vec{c} \leq \vec{c^+}$, which means that the accelerated time increase went over the original cut.
\end{itemize}

The algorithm is described in python : $list.add(e)$ adds element $e$ in the queue $list$ (only if $e \notin list$) , 
while $list.pop()$ removes the first element (it is a first in/first out behaviour), $border$ and $exploring$ are initially empty and the loop condition is true as long as $exploring$ is nonempty.

\begin{algorithm}
\caption{$next\_border(\state)$}
\label{algo:next_border}
\begin{algorithmic}
\STATE{$border[]$} \COMMENT{Set of states to be returned}
\STATE{$exploring[]$} \COMMENT{Queue of states to explore}
\STATE{$exploring.add(\state$)}
\WHILE{$exploring$}
	\STATE{$pre\_s \leftarrow exploring.pop()$} 
	\STATE{$successors \leftarrow next\_state(pre\_s)$} 
	\FORALL{$s \in successors$}
		\IF{$is\_cut(pre\_s,s)$} 
			\STATE{$border.add(s)$} \COMMENT{States of the cut are added to border}
		\ELSE
			\STATE{$exploring.add(s)$} \COMMENT{Other states are added to exploring}
		\ENDIF
	\ENDFOR
\ENDWHILE
\RETURN $border$
\end{algorithmic}
\end{algorithm}

This can be used iteratively in a depth-first exploration to jump from a state to one of its successors belonging to the next border.
During this exploration, an additional function may be used to check if a state satisfies some condition. 
Such a use of layers allows to reduce the number of explored paths by detecting 
diamonds caused by the order of transitions of concurrent agents.

\subsection{Exploration using strong and weak variables}

The approach presented in the previous section does not deal with diamonds spreading on a time distance 
longer than the one between two adjacent borders.
For example, it may still happen that two different states $s_1$ and $s_2$ belonging to the same border have a common successor $s$ in the future, as illustrated in Fig.~\ref{fig:layered_space}. 
To cope with this issue, it is more interesting to perform the width-first exploration that computes successors at the next border for the set $\{s_1, s_2\}$ instead than taking them separately.
In general, it is not obvious to know or guess 
which states should be kept together in the computation of the next border. 
Indeed, one should be able to determine when sets of states should be split in sub-sets and when they should be kept together.
To perform such a clustering, it may be interesting to exploit the properties of target applications, 
such as \cavs.

A possible solution is to assume $V \defeq V_w \times V_s$, where $V_w$ (weak) is a less important part of $V$ and $V_s$  (strong) a more important one, such that states differing in the valuation of $V_s$ are unlikely to have a common successor, while this is not the case for $V_w$.
Symmetrically, states with the same 
valuation of $V_s$ are more likely to have a common successor.
This may give us a criterion to cluster states and jump from a set of states to the set of their successors at the next border.
The choice of $V_s$ and $V_w$ is of course system-dependent and should be defined by an expert, or with the help of a simulation tool.
As an example, elements that can be assigned a new value independently of their previous one might be considered as weak, while elements whose value changes depend on their present value (for instance the position of a moving object) might be considered as strong.

Function $clustered\_next\_border(\state\_set)$ is then a variant of $next\_border()$,
taking a set of states and producing a set of clusters, \ie, sets of states having 
identical values of variables in $V_s$.
It is used in a similar way as $next\_border()$ to explore in a depth-first manner
the layered state space, the only difference being that it jumps from a cluster belonging to some border to a cluster belonging to the next one,
based on the choice of $V_s$. 

\drop{
The exploration taking into account the structure of $V$ is defined below.
Formally, we implement a function $strong\_equal(s1,s2)$ checking equality on $V_s$ that returns $\true$ if the valuation of $V_s$ is identical for states $s1$ and $s2$.
Then we modify $next\_border()$ as defined in Algorithm~\ref{algo:next_border2} so it now takes a set of states $\state\_set$ as an input and a set of sets of states as an output.

\begin{algorithm}
\caption{$next\_border(\state\_set)$}
\label{algo:next_border2}
\begin{algorithmic}
\STATE{$border[][]$}
\STATE{$exploring[]$}
\STATE{$state\_added$}
\FORALL{$\state \in \state\_set$}
    \STATE{$exploring.add(\state$)}
\ENDFOR
\WHILE{$exploring$}
	\STATE{$pre\_s \leftarrow exploring.pop()$} 
	\STATE{$successors \leftarrow next\_state(pre\_s)$} 
	\FORALL{$s \in successors$}
		\IF{$is\_cut(pre\_s,s)$}
		    \STATE{$state\_added \leftarrow False$}
		    \FORALL{$set \in border$}
                \IF{$strong\_equal(set[0],s)$}
                    \STATE{$set.add(s)$} 
                    \STATE{$state\_added \leftarrow \true$}
                \ENDIF
            \ENDFOR
		    \IF{$\lnot state\_added$}
			    \STATE{$border.add([s])$}
			\ENDIF
		\ELSE
			\STATE{$exploring.add(s)$}
		\ENDIF
	\ENDFOR
\ENDWHILE
\RETURN $border$
\end{algorithmic}
\end{algorithm}
}

Note that if $V_s = \emptyset$, such an exploration is equivalent to a classical width-first one, since states at a border are always kept is the same sub-set.
With such an algorithm, for a bounded layered state space of a \model, 
one can perform an "on-the-fly" depth-first exploration since there is no need to memorize explored states.
This may be used to efficiently search for specific reachable states, and may be sped up by the use of heuristics that choose which sets of states to explore first.

\section{Dynamic exploration of a \model{}}
\label{sec:explo}

This section is dedicated to exploration algorithms of finite prefixes of \model{}s:
states that do not have successors in the considered prefix will be called \emph{final}.
The algorithms are denoted with the \ctl{} temporal logic syntax.
Since this temporal logic is meant to explore infinite paths, 
we shall consider that each final state has a self loop.

Our algorithms have two main characteristics: they operate "on-the-fly", which means that 
they do not store the entire visited state space (but only a cut of it), and they can be tuned with heuristics defining a priority on paths to be explored, that might significantly speed up the computation time if the searched states exist.
To do so we rely on the algorithm $clustered\_next\_border()$ mentioned 
in Section~\ref{layers}.
Since they do not store all the states that have been explored, we chose
not to return traces of execution, unlike what is usually proposed by standard temporal logic model checking tools.

We formalise in the following 
algorithms for the basic \ctl{} properties $EF p$ and $EG p$, respectively meaning \textit{a reachable state satisfies $p$} and \textit{there exists a path where $p$ is always true}. 
Any property for which we have an algorithm may be negated, so that we can also express $AF p$ and $AG p$, respectively equivalent to $\neg(EG \neg p)$ and $\neg(EF \neg p)$.

The algorithm for $EF p$ consists, starting from a stack containing the initial state, in taking the first element $s$ of the stack, returning it if $p$ is true on $s$, and otherwise adding the result of function $clustered\_next\_border(s)$ to the stack.
The algorithm continues recursively until reaching $p$ or there is no more states to explore in the considered finite prefix.
Additionally, we return $true$ if $p$ is satisfied by a state  
between two borders, \ie, during an application of $clustered\_next\_border()$. 

The algorithm for $EG p$ works in a similar way, but the state $s$ is returned if $p$ is true on $s$ and if $s$ is final, and $clustered\_next\_border(s)$ is added to the stack only if $p$ is true on $s$.
Additionally, states where $p$ is not true are dropped when 
explored in $clustered\_next\_border()$. 
That way, only states where $p$ is true are explored.

We may also define algorithms for nested \ctl{} queries built with binary logical operators. We shall for example consider 
two of them: $EF (p \land EF q)$, meaning that \textit{a reachable state satisfies $p$ and from that state a reachable state satisfies $q$}, and $EF (p \land EG q)$, meaning that \textit{a reachable state satisfies $p$ and from that state there exist a path where $q$ is always true}.
One may notice that the "leads to" operator ($\fleche$) used in the state of the art tool \uppaal{} follows the equivalence : 
$p \fleche q <=> AG (\neg p \lor AF q) <=> \neg EF (p \land EG \neg q)$.
This operator is therefore expressible in our framework.
Although only these two queries are given here, any kind of nested \ctl{} query can be implemented.

Those nested queries are implemented using a marking function (\ie, a Boolean indicator).
$EF (p \land EF q)$ is implemented as follows.
Whenever $p$ is true on a state, the state is marked.
Whenever a state is marked, all its successors are marked.
Starting from a stack containing the initial state, the first element $s$ of the stack is returned if $q$ is true on $s$ and $s$ is marked.
Otherwise, the result of $clustered\_next\_border(s)$ is added to the stack.
The same marking process is performed between two borders, \ie, in $clustered\_next\_border()$. 
We continue recursively until a state validates the property or there is no more state to explore.
As for $EF (p \land EG q)$, states are marked whenever both $p$ and $q$ are true or the state is a successor of a marked state and $q$ is true.
If a marked final state is reached, it validates the property and is returned.
Again, the same marking process is performed in $clustered\_next\_border()$. 

\section{Experiments}
In this section we illustrate the performances of our exploration algorithms.
To do so, we use \model{}s representing systems of autonomous communicating vehicles, for which both Constraints  
\ref{constr:liveness} and \ref{constr:acyclic} are satisfied.
The first constraint 
allows to use the acceleration, which heavily reduces the size of the state space as well as the number of diamonds.
The second constraint 
ensures that the state space is a \acygra{}.
As a consequence of the latter, 
the state space is infinite, because of the $X$ part of $V$.
In the following case studies, the longitudinal positions of the vehicles on the road will play the role of this part.
The road we observe is technically infinite, but as we are interested only in the analysis of a portion of it, we can bound the exploration to a fixed value of $X$.
The system thus converges towards a bound that, once reached, is considered as a final state.

In the following, we first compare the exploration time obtained with or without acceleration. 
Then, we discuss the advantages and drawbacks of using various types of layer-based explorations.
In the third part of this section, we provide some heuristics, and experiment them in order to (hopefully) observe the gain that can be achieved with them. 
Finally, we compare this method with the framework \verifcar{} \cite{verifcar}, which uses \uppaal, 
and we provide a verification method for the analysis of such  systems that is more efficient than the one proposed in \cite{verifcar}. 

Three models used in~\cite{verifcar}, featuring various 
state space sizes, have been implemented as \model{}s. 
Those models represent systems of autonomous vehicles circulating on a portion of highway where each vehicle communicates with the other ones to make decisions about its behaviour. 
These experiments have been performed by implementing the models with the free high level Petri net tool \zinc{}, 
using its library to implement our exploration algorithms.

\subsection{Efficiency of the accelerated dynamics.} 

A width first exploration of the state space on each of the three models have been performed using both the original and the accelerated semantics.
Table~\ref{comp_accel} provide for each model the number of states in its state space along with their full exploration times (FET) in both semantics and in \uppaal{}.
As expected, the accelerated semantics reduced exploration time, 
therefore, it has been used in all the subsequent experiments.

\begin{table}[htb]
\centering
\arraycolsep=5pt
$
\begin{array}{|l|c|c|c|}
\hline
&\mbox{Model 1}   &\mbox{Model 2}   &\mbox{Model 3}\\
\hline
\hline
\mbox{FET original semantics (s) }  &14.5   &144   &574\\
\hline
\mbox{FET accelerated semantics (s) }  &10.8   &52.1   &420\\
\hline
\mbox{FET \uppaal{} (s) } &5     &36   &379\\
\hline
\mbox{Size of the state space} &7751   &52732   &285944 \\
\hline
\end{array}
$
\vspace{.2cm}
\caption{\label{comp_accel}}
\end{table}

It is interesting to mention that the main improvement of the accelerated semantics, compared to the original one, is to explore only one state of each (maximal) action zone.
As such, the more a system features transitions with wide non-deterministic time intervals, the greater is the time gain provided by the accelerated semantics.
Here, the non-deterministic time intervals present in Model 1 and Model 3 are quite short, such that the number of paths that are ignored in the accelerated dynamics is not very important.
On the other hand, Model 2 features a transition with a wider non-deterministic time interval, explaining why the difference between the two semantics is more pronounced for this model.
We can thus expect 
the accelerated semantics to be even more useful when using models similar to the one depicted in Figure~\ref{fig:intervals}.

\subsection{Efficiency of the layer-based exploration.}

Here, we compare, for several exploration algorithms, 
the full exploration time and the reachability time of the first occurrence of a final state. 
They are explored in width first, depth first without layers and depth first with layers (with and without the use of strong/weak variables). The size of the list $\mathit{Cuts}$ was 1 for the models 1 and 3, and 5 for the second one. 
Table~\ref{comp_algo} shows the results.

\begin{table}[htb]
\centering
\resizebox{\columnwidth}{!}{
\arraycolsep=5pt
$
\begin{array}{|l|c|c|c|c|c|c|c|}
\hline
\mbox{Exploration algorithm}    
&\multicolumn{3}{|c|}{\mbox{Full exploration time (s)} }	&   &\multicolumn{3}{c|}{\mbox{First occurrence of a final state (s)}}\\
\cline{2-4}\cline{6-8}
&\mbox{Model 1}   &\mbox{Model 2}   &\mbox{Model 3}   & &\mbox{Model 1}   &\mbox{Model 2}   &\mbox{Model 3}\\
\cline{1-4}\cline{6-8}
\mbox{Width first}  &10.8   &52.1   &420  &  &10.7   &52   &419.9\\
\hline
\hline
\mbox{Depth first without layers}  &\infty   &\infty   &\infty  &  &3.3   &4.6   &3.9\\
\cline{1-4}\cline{6-8}
\mbox{Depth first layered ($V_w = \emptyset$)}  &11   &\infty   &\infty  &  &4.5   &6.6   &4.1\\
\cline{1-4}\cline{6-8}
\mbox{Depth first layered (small $V_w$)}  &11   &250   &2015  &  &4.5   &7.4   &6\\
\cline{1-4}\cline{6-8}
\mbox{Depth first layered (large $V_w$)}  &11   &71    &667  &  &4.5   &14.4   &7.2\\
\hline
\end{array}
$
}
\vspace{.2cm}
\caption{\label{comp_algo} Comparison of full exploration time and time to reach the first occurrence of a final state state for exploration algorithms in width first, depth first with and without layers and with or without the use of weak variables. $\infty$ means that the exploration was stopped after $50$ hours of computation without a result.}
\end{table}

One can see that the width first algorithm has the best full exploration time in any case,
but the time before reaching any final state is close to the full exploration one, which makes it the worst technique in this case.
On the other hand, the standard depth first algorithm is the fastest for reaching a final state, but it does not fully explore the state space even after $50$ hours of computation.

Results for Model 1 show that as long as the layer based approach is used, the full exploration time is very close to that of the width first algorithm.
This indicates that there is almost no diamonds covering several layers, 
meaning that different states belonging to the border of a layer almost never share a common successor.
Because of that, the use of weak variables has no effect.
Although this case is rather simple, it clearly highlights the advantage of layer-based exploration: with almost no increase in full exploration time, it is able to reach a final state much faster.

Model 2 and Model 3 
have much more complex state spaces and,
in these cases, the layer-based algorithm that does not rely on weak variables to aggregate states is not able to explore the full state space even after $50$ hours of computation.
On the contrary, 
using even a small number but well chosen 
weak variables ($6$ out of $39$), it is possible to fully explore the state space. In both cases, exploration is about five times longer than the exploration time of width first algorithm.
When using a large number of weak variables ($30$ out of $39$), 
the exploration is much shorter (about $1.5$ times the time of width first algorithm).
One can note however that the larger $V_w$, the longer it takes to reach a final state.
Indeed, as states are aggregated layer by layer, a too large $V_w$ would result in an exploration similar to a width first one, where all states are kept together and final states are only reached at the end of the computation.
With the weak variables chosen, the time to reach a final state remains however reasonable. 

In the next experiments, the depth first algorithms always use layers and a fixed non-empty $V_w$.

\subsection{Heuristics}

Exploration algorithms based on layers allow the use of heuristics.
These heuristics guide the exploration, choosing among all the unexplored states the one that will most likely lead to a state that satisfies the checked property.
The heuristics we use consists in associating a weight to each state.
When a new state is discovered, it is placed in a list  ordered by weight of states to explore.
The list of states to explore is sorted either by ascending or descending weight, depending on the property to verify.
The weight is a prediction of the distance between the current state and a state satisfying the property. The next state to be explored is the last in the list, \ie, having the highest (respectively lowest) weight.

Therefore, a property may be associated with a heuristics that takes a state as an input and returns a weight as an output.
Below is a list of heuristics that we used for experiment purposes together with the property they are associated to:
\begin{enumerate}
    \item $\mathit{distance\_vh_1\_vh_2}$: returns the longitudinal position of vehicle $vh_1$ minus that of vehicle $vh_2$.
    It may be used with property $EF\ arrival\_vh_1\_before\_vh_2$ and weights sorted in ascending order, where $\mathit{arrival\_vh_1\_before\_vh_2}$ is true in a state if vehicle $vh_1$ reaches the end of the road portion before vehicle $vh_2$ does.
    The idea behind is to check in priority states where $vh_1$ is the most ahead of $vh_2$.
    \item $\mathit{estimated\_travel\_time\_vh}$: returns the time traveled since the initial state plus the estimated time to reach the end of the road portion, assuming the current speed is maintained. 
    It may be used with weights sorted in ascending order and property $EF\ travel\_time\_vh\_sup\_n$, where $\mathit{travel\_time\_vh} \geq n$ is true in a state if $vh$ has reached the end of the road portion within $n$ time units.
    The idea is to check in priority states where $vh$ is predicted to reach the end of the road with the shortest time. 
    \item $\mathit{time\_to\_overtake\_vh_1\_vh_2}$: is the time before both vehicles arrive at the same longitudinal position if they keep their current speed.
    It may be used with weights sorted in descending order and property $EF\ ttc\_vh_1\_vh_2 \leq n$, where $\mathit{ttc\_vh_1\_vh_2}$ is the value of the time to collision indicator between $vh_1$ and $vh_2$ (\ie, the delay before there is a collision between the two vehicles if they keep their current speed), and $n$ is a time to collision value.
    The idea is to check in priority states where one of the vehicles is getting closer to the other one with the higher speed.
\end{enumerate}

These heuristics have been used on Model $3$, with results given in Table \ref{comp_heuri}.
The scenario in Model 3 considers three vehicles positioned as depicted in Fig.~\ref{fig:scenario} on a two lane road portion that is 500 m long, with one additional junction lane.
Initially, vehicle A is on the right lane at position 0 m with a speed of 30 m/s, vehicle B is on the left lane at position 30 m with a speed of 15 m/s and vehicle C is on the junction lane at position 40 m with a speed of 20 m/s.
They all aim at being on the right lane at the end of the road portion.
 
\begin{figure}[ht]
\centering
\includegraphics[scale=0.4]{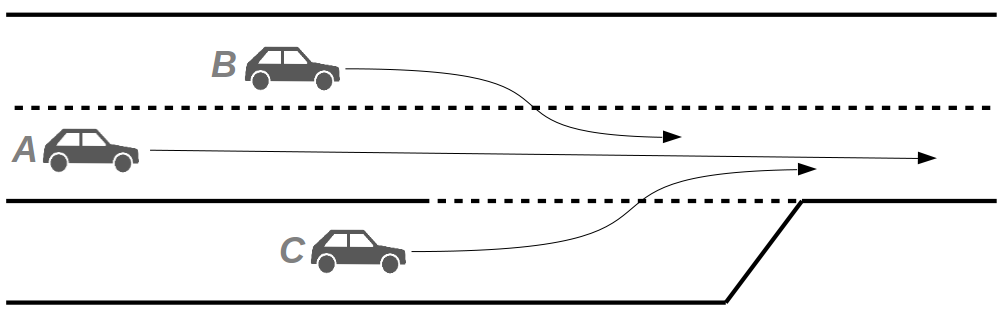}
\caption{\label{fig:scenario} Initial positions and possible trajectories of autonomous vehicles for the scenario in Model $3$.}
\end{figure}

\begin{table}[htb]
\centering
\resizebox{\columnwidth}{!}{
\arraycolsep=5pt
$
\begin{array}{|l|c|c|c|c|}
\hline
\mbox{Exploration algorithm}    &\mbox{$EF\ arrival\_B\_before\_A$}   &\mbox{$EF\ travel\_time\_A \geq 15.9$}   &\mbox{$EF\ ttc\_A\_C \leq 1.14$}   &\mbox{$EF\ ttc\_A\_B \leq 0$}\\
\hline
\mbox{Width first}  &416  &427  &292    &95\\
\hline
\mbox{Depth first without heuristics}  &234\mbox{---}357  &167\mbox{---}340  &247\mbox{---}547    &277\mbox{---}483\\
\hline
\mbox{Depth first with heuristics}  &131  &149  &103 &13\\
\hline
\end{array}
$
}
\vspace{.2cm}
\caption{\label{comp_heuri} Comparison of reachability time for exploration algorithms in width first and depth first with and without heuristics.
As depth first without heuristics is non deterministic, the two values correspond to the fastest and the slowest runs obtained for each query (five runs where performed each time).
}
\end{table}

The first two queries can only be true in a final state (the deepest layer).
As such, the reachability time with the width first algorithm is close to the full exploration time with the same algorithm.
In general, the width first reachablity time depends on the depth of the first state that satisfies the property.
One can observe that for the fourth query, the state is actually reached at a lower depth, which is reflected by the reachability time. 

As the depth first algorithm without heuristics randomly chooses which paths to explore first, the reachability time varies.
The number of states in the whole state space that satisfies the property thus impacts the mean reachability time with this algorithm, \ie, when there is more possibility to verify the property, the average time needed is shorter.
As we do not want to rely on luck, this is not satisfying.

On the other hand, depth first algorithm with heuristics explores states in a given order (depending on their weights) and therefore the reachablity time is always the same.
The heuristics we used could of course be modified and improved, but they are enough to show a significant decrease of 
the reachability time.
Even for the fourth query, where the width first is faster than the depth first algorithm, the heuristics allows to quickly identify the state that satisfies the property.

\subsection{Comparison with \verifcar{}}

We will now compare the reachability time obtained with \uppaal{} with the ones obtained with the 
depth first exploration algorithms with heuristics, on Model $3$.
We observed that \uppaal{} first constructs the state space in about $106$ s, then is able to answer almost instantly if a searched state exists.
It can therefore answer several queries after constructing the state space, unlike
our heuristics-based  dynamic exploration algorithms, which have to explore the state space from scratch for each query.
Yet, most of the states we aimed for can be reached easily, and the computation took only about $4$ seconds.
Queries depicted in Table \ref{comp_heuri} are those where states were harder to reach.
Compared to the ones we obtained in~\cite{verifcar}, 
these results indicate that, when a reachability property is verified, our algorithms have the same kind of execution time than the ones observed with \uppaal{}. 
On the other hand, if the reachability property is not true, they are slower than \uppaal{}, which depending on the kind of query takes between $34$ and $370$ seconds, which equals the full exploration time with this tool for Model 3.
As mentioned previously, the full state space exploration time with depth first algorithms on this Model, is in our case, of $667$ seconds.
This is not a surprise since \uppaal{} is a mature tool using many efficient abstractions.

However, \uppaal{} is restricted in terms of expressivity, at least in two ways interesting for us.
First, it is not possible to directly check bounds of a given numerical indicator, and such bounds should be obtained by dichotomy, requiring several runs for each indicator, such as proposed in the methodology of~\cite{verifcar}.
Second, it is limited to a subset of CTL (accepting mainly non nested queries).
Our algorithms do not have such restrictions.

Indeed it is possible to do a full exploration of the state space while keeping, for each state, the lower and higher values 
reached on the paths leading to the state, 
for a given set of indicators.
That way, all the information needed to analyse the behaviour of the system, can be obtained after only one full exploration. 
This is performed as a standard width-first exploration (storing states in a file) with the difference that each state is associated to a set of pairs $(min,max)$, being the (temporary) bounds of the considered indicators.
Each time a state is explored, the value for each indicator is computed, and it overwrites $min$ if it is smaller, and $max$ if it is greater.
That way, each state $s$ contains, for each indicator, the smallest and highest values that exist on the paths from the initial state to $s$.
As several paths can lead to $s$, we will consider that $s$ reached from path $P1$ and $s$ reached from path $P2$ are equivalent only if the set of their indicators are also equivalent.
Therefore, some diamonds might be detected (\ie, two identical states coming from different paths) but not merged together in order to keep information about their respective paths. 
That way it is possible to have several versions of the same state, but with different indicator values. 
If one is interested in the reachability of states (for instance, if an indicator is equal to some value), 
this can easily be done in the same way, by adding Boolean variables to the set of indicators.
At the end of the exploration, we get this way the set of all final states, together with all the information that has been carried on their respective paths.
It therefore contains all the information needed to analyse finely the system.
For the case of Model 3, getting the arrival order together with the bounds for travel time and worst time-to-collision takes $708$ seconds.
In comparison, the time needed by \uppaal{} to obtain the same information with the dichotomy procedure is $3553$ seconds.

Also, the \acygra{} shape of the state space allows us to implement any kind of \ctl{} queries.
For the experiments, we used a query of the kind $EF (p \land EG q)$, which is the negation of the "leads to" operator $p \fleche \neg q$ 
(the only nested operator available in \uppaal, in addition to deadlock tests)
and two of the kind $EF (p \land EF q)$, which cannot be expressed in \uppaal{}. 

In \cite{verifcar}, $arrival\_C\_before\_A \fleche arrival\_B\_before\_A)$ was used and reached a state invaliding the property in $110$ seconds.
Its negation can be expressed here as 
$EF (arrival\_C\_before\_A \land EG \neg arrival\_B\_before\_A)$
and our algorithm finds the state validating the property in about $10$ seconds.
The properties $EF (ttc\_A\_B \leq 1 \land EF arrival\_A\_before\_C$ and $EF (ttc\_A\_B \leq 1 \land EF arrival\_A\_before\_B)$, 
that cannot be checked in \uppaal{}, can be expressed here.
The first one expresses the possibility for vehicle A to arrive ahead of vehicle C after a dangerous situation has occurred, involving a time to collision of less than 1 second. The second is similar for vehicle A and B in the same conditions.
The first query is false and needs to explore the whole state space to give an answer (in $680$ seconds), while the second one is true and finds a state satisfying the property in about $10$ seconds.

Finally, it is worth mentioning that discretisation is needed for \uppaal, and therefore approximations may be mandatory in some cases, leading to a loss in precision and realism.
In addition to a better expressivity, the model checking process presented in this paper also ensures that no approximation is needed, hence a higher level of realism is achieved.

\section{Conclusion}

We introduced \gmodel{}s, multi-agent timed models where each agent is associated to a regular timed schema 
upon which all possibles actions of the agent rely.
The formalism allows to easily model systems featuring a high level of concurrency between actions, where actions are not temporally deterministic, such as the \cav{}s.
We have then formalised \model{}s ({\em Multi-Agent with timed Periodic Tasks}), by soundly constraining \gmodel{} ones.
\model{}s allows for an accelerated semantics 
which is an abstraction that greatly reduces the size of the state space by reducing as much as possible the number of time passings in the system.
We also presented how to extract a layered structure out of a \model{}, that allows to detect diamonds while exploring the system depth first.
We provided a translation from \gmodel{} to high level Petri nets, which allowed us to 
implement a dedicated checking environment for this formalism with the (free) academic tool \zinc{}.
Algorithms implemented in such environments explore state spaces dynamically and can be used together with heuristics that allow to reduce 
the computation time needed to reach some states in the model.
Finally, experiments highlighted the efficiency of our abstractions, and a comparison of model checking \cav{}s systems with the framework \verifcar{} has been performed.
Although our checking environment does not return traces of executions and is not better for full exploration times than the state of art tool \uppaal{} used in \verifcar{}, it has a better expressivity both on the model, since we can compute with non-integer 
numbers, and on the queries since nested \ctl{} ones can be checked. 
The heuristics performed well for reachability problems and we also provided an exploration algorithm that allows to gather all information needed to analyse the system in one run, which greatly decreased the time needed to gather the same amount of information when using \verifcar{}.
Although we developed this method with the case study of autonomous vehicles in mind, this formalism and all the abstractions and algorithms presented in this paper can be easily applied to any multi-agent real time systems where agents adopt a cyclic behaviour,
such as mobile robots completing cyclically tasks according to their own objectives, flying drone squadrons, etc.

\bibliographystyle{plain}
\bibliography{bib.bib}

\begin{thebibliography}{10}

\bibitem{AlurDill90}
R.~Alur and D.~Dill.
\newblock Automata for modelling real-time systems.
\newblock In {\em Proceedings of the International Colloquium on Automata,
  Languages and Programming (ICALP'90)}, volume 443 of {\em LNCS}, pages
  322--335. Springer-Verlag, 1990.

\bibitem{verifcar}
Johan Arcile, Raymond Devillers, and Hanna Klaudel.
\newblock Verifcar: a framework for modeling and model checking communicating
  autonomous vehicles.
\newblock {\em Autonomous Agents and Multi-Agent Systems}, 33(3):353--381, May
  2019.

\bibitem{Biere2003}
Armin Biere, Alessandro Cimatti, Edmund~M Clarke, Ofer Strichman, Yunshan Zhu,
  et~al.
\newblock Bounded model checking.
\newblock {\em Advances in computers}, 58(11):117--148, 2003.

\bibitem{Clarke2001}
Edmund Clarke, Armin Biere, Richard Raimi, and Yunshan Zhu.
\newblock Bounded model checking using satisfiability solving.
\newblock {\em Formal Methods in System Design}, 19(1):7--34, Jul 2001.

\bibitem{DBLP:series/eatcs/Jensen92}
Kurt Jensen.
\newblock {\em Coloured Petri Nets - Basic Concepts, Analysis Methods and
  Practical Use - Volume 1}.
\newblock EATCS Monographs on TCS. Springer, 1992.

\bibitem{kong15}
S.~Kong, S.~Gao, W.~Chen, and E.~Clarke.
\newblock dreach: $\delta$-reachability analysis for hybrid systems.
\newblock In Christel Baier and Cesare Tinelli, editors, {\em Tools and
  Algorithms for the Construction and Analysis of Systems}, pages 200--205,
  Berlin, Heidelberg, 2015. Springer Berlin Heidelberg.

\bibitem{larsen:97}
Kim~G. Larsen, Paul Pettersson, and Wang Yi.
\newblock Uppaal in a nutshell.
\newblock {\em International Journal on Software Tools for Technology Transfer
  (STTT)}, 1(1-2):134--152, Oct 1997.

\bibitem{OKelly201616AA}
M.~O'Kelly, H.~Abbas, and R.~Mangharam.
\newblock {APEX} : Autonomous vehicle plan verification and execution.
\newblock In {\em SAE World Congress}, 2016.

\bibitem{peterson}
James~L. Peterson.
\newblock {\em Petri Net Theory and the Modelling of Systems}.
\newblock Prentice Hall, 1981.

\bibitem{platzer09}
A.~Platzer and J.-D. Quesel.
\newblock European train control system: A case study in formal verification.
\newblock In Karin Breitman and Ana Cavalcanti, editors, {\em Formal Methods
  and Software Engineering}, pages 246--265, Berlin, Heidelberg, 2009. Springer
  Berlin Heidelberg.

\bibitem{pommereau:hal-01941485}
Franck Pommereau.
\newblock {ZINC: a compiler for ``any language''-coloured Petri nets}.
\newblock Technical report, {IBISC, university of Evry / Paris-Saclay}, 2018.

\bibitem{quottrup04}
M.~M. Quottrup, T.~Bak, and R.~I. Zamanabadi.
\newblock Multi-robot planning : a timed automata approach.
\newblock In {\em IEEE International Conference on Robotics and Automation,
  2004. Proceedings. ICRA '04. 2004}, volume~5, pages 4417--4422 Vol.5, April
  2004.

\bibitem{Sorea2003}
Maria Sorea.
\newblock Bounded model checking for timed automata.
\newblock {\em Electronic Notes in Theoretical Computer Science},
  68(5):116--134, 2003.

\bibitem{uppaal4.1}
Uppaal.
\newblock http://www.uppaal.org/.

\end{thebibliography}

\end{document}